\documentclass[11pt,a4paper]{article}

\usepackage{fancyhdr}
\usepackage{amsmath}
\usepackage{bm}
\usepackage{esint}
\usepackage{amsfonts}
\usepackage{amsthm}
\usepackage{amssymb}
\usepackage{amsbsy}
\usepackage{verbatim}
\usepackage{graphicx}
\usepackage{url}
\usepackage{mathrsfs}
\usepackage[hmargin = 1in,vmargin=.75in]{geometry}
\usepackage{color}
\usepackage{amstext}
\usepackage{stmaryrd}
\usepackage{enumerate}
\usepackage{algpseudocode}
\usepackage{algorithm}
\usepackage{pifont}
\usepackage{prettyref}
\usepackage{hyperref}
\hypersetup{
    colorlinks = true,
}
\font\msbm=msbm10

\numberwithin{equation}{section}

\theoremstyle{plain}
\newtheorem{theorem}{Theorem}[section]
\newtheorem{lemma}[theorem]{Lemma}
\newtheorem{corollary}[theorem]{Corollary}

\def\mathbb#1{\hbox{\msbm{#1}}}

\newcommand{\be}{\boldsymbol{e}}

\newcommand{\bu}{\boldsymbol{u}}

\newcommand{\bx}{\boldsymbol{x}}

\newcommand{\bone}{\boldsymbol{1}}

\newcommand{\BPhi}{\boldsymbol{\Phi}}
\newcommand{\BA}{\boldsymbol{A}}
\newcommand{\BB}{\boldsymbol{B}}

\newcommand{\BG}{\boldsymbol{G}}

\newcommand{\BJ}{\boldsymbol{J}}
\newcommand{\BL}{\boldsymbol{L}}

\newcommand{\BO}{\boldsymbol{O}}

\newcommand{\BQ}{\boldsymbol{Q}}
\newcommand{\BR}{\boldsymbol{R}}
\newcommand{\BS}{\boldsymbol{S}}

\newcommand{\BU}{\boldsymbol{U}}
\newcommand{\BV}{\boldsymbol{V}}
\newcommand{\BW}{\boldsymbol{W}}
\newcommand{\BX}{\boldsymbol{X}}
\newcommand{\BY}{\boldsymbol{Y}}
\newcommand{\BZ}{\boldsymbol{Z}}

\newcommand{\BDelta}{\boldsymbol{\Delta}}
\newcommand{\BPi}{\boldsymbol{\Pi}}

\newcommand\keywords[1]{\textbf{Keywords}: #1}
\newcommand\MSC[1]{\textbf{MSC numbers}: #1}

\newcommand{\BLambda}{\boldsymbol{\Lambda}}

\newcommand{\PP}{\mathcal{P}}

\newcommand{\pa}{\partial}

\newcommand{\I}{\boldsymbol{I}}
\newcommand{\RR}{\mathbb{R}}
\newcommand{\lag}{\langle}
\newcommand{\rag}{\rangle}

\newcommand{\eps}{\epsilon}

\newcommand*\diff{\mathop{}\!\mathrm{d}}
\DeclareMathOperator{\Tr}{Tr}

\DeclareMathOperator{\Od}{O}

\DeclareMathOperator{\E}{\mathbb{E}}

\DeclareMathOperator{\diag}{diag}
\DeclareMathOperator{\blkdiag}{blkdiag}

\DeclareMathOperator{\BDG}{BDG}

\DeclareMathOperator{\St}{St}

\renewcommand{\Pr}{\mathbb{P}}

\definecolor{xl}{RGB}{200,50,120}

\begin{document}
\title{\bf Local Geometry Determines Global Landscape in Low-rank Factorization for Synchronization}

%\author{Shuyang Ling\thanks{Shanghai Frontiers Science Center of Artificial Intelligence and Deep Learning, New York University Shanghai. Address: 567 Yangsi Road, Pudong New District, Shanghai, China, 200124. S.L. is (partially) financially supported by the National Key R\&D Program of China, Project Number 2021YFA1002800, National Natural Science Foundation of China (NSFC) No.12001372, Shanghai Municipal Education Commission (SMEC) via Grant 0920000112, and NYU Shanghai Boost Fund.} }

\author{Shuyang Ling\thanks{Shanghai Frontiers Science Center of Artificial Intelligence and Deep Learning, New York University Shanghai. Address: 567 Yangsi Road, Pudong New District, Shanghai, China, 200124. S.L. is (partially) financially supported by the National Key R\&D Program of China, Project Number 2021YFA1002800, Shanghai STCSM Rising Star Program No. 24QA2706200, STCSM General Program No. 24ZR1455300, National Natural Science Foundation of China No.12001372, Shanghai SMEC No. 0920000112, SMEC AI Initiative Program, and NYU Shanghai Boost Fund.}}

\maketitle

\begin{abstract}
The orthogonal group synchronization problem, which focuses on recovering orthogonal group elements from their corrupted pairwise measurements, encompasses examples such as high-dimensional Kuramoto model on general signed networks, $\mathbb{Z}_2$-synchronization, community detection under stochastic block models, and orthogonal Procrustes problem.
The semidefinite relaxation (SDR) has proven its power in solving this problem; however, its expensive computational costs impede its widespread practical applications. We consider the Burer-Monteiro factorization approach to the orthogonal group synchronization, an effective and scalable low-rank factorization to solve large scale SDPs. Despite the significant empirical successes of this factorization approach, it is still a challenging task to understand when the nonconvex optimization landscape is benign, i.e., the optimization landscape possesses only one local minimizer, which is also global.  In this work, we demonstrate that if the rank of the factorization exceeds twice the condition number of the ``Laplacian" (certificate matrix) at the global minimizer, the optimization landscape is absent of spurious local minima. Our main theorem is purely algebraic and versatile, and it seamlessly applies to all the aforementioned examples: the nonconvex landscape remains benign under almost identical condition that enables the success of the SDR. 
Additionally, we illustrate that the Burer-Monteiro factorization is robust to ``monotone adversaries", mirroring the resilience of the SDR. In other words,  introducing ``favorable" adversaries into the data will not result in the emergence of new spurious local minimizers. 
\end{abstract}

\vskip0.4cm
\MSC{90C26, 90C30, 90C35, 90C46, 34C15}

\section{Introduction}

Consider the orthogonal group synchronization:
\[
\BA_{ij} = \BO_i\BO_j^{\top} + \text{noise},
\]
where $\BO_i$ belongs to $\Od(d)$, defined as:
\[
\Od(d) = \{ \BG\in\RR^{d\times d}: \BG\BG^{\top} = \BG^{\top}\BG=\I_d \}
\]
This set includes all $d\times d$ orthogonal matrices.
It is more appropriate to treat $\BA_{ij}$ as a noisy version of $\BO_i\BO_j^{\top}$ subject to particular corruption models.
The goal of orthogonal group synchronization is to estimate the orthogonal matrices $\{\BO_i\}_{i=1}^n$ from their incomplete and corrupted pairwise measurements $\BA_{ij}$. 
The $\Od(d)$ synchronization has gained significant attention, encapsulating numerous problems emerging from data science and machine learning such as community detection~\cite{ABH15}, computer vision~\cite{HG13}, and cryo-EM~\cite{S18}. It also serves as a crucial subtask in special Euclidean group synchronization arising from robotics~\cite{RCBL19}. Moreover, it shares close ties with the synchronization phenomenon observed in the studies on collective behaviors of coupled oscillators~\cite{ABK22,K75,MTG17,MTG20,LXB19,MB23}. 

\paragraph{Least squares estimation and convex relaxation}
One common strategy is to identify the least-squares estimator:
\begin{equation}\label{def:ls}
\min_{\BR_i\in\Od(d) }~\sum_{i<j} \|\BR_i\BR_j^{\top} - \BA_{ij}\|_F^2 \Longleftrightarrow \min_{\BR_i\in\Od(d) }~-\lag \BA,\BR\BR^{\top}\rag
\end{equation}
where $\BR^{\top} = [\BR_1^{\top},\cdots,\BR_n^{\top}]\in\Od(d)^{\otimes n}$ {is a 
concatenation of $n$ $d\times d$ orthogonal matrices} and the $(i,j)$-block of $\BA$ is $\BA_{ij}.$ {For $d=1$,~\eqref{def:ls} is essentially the graph max-cut problem and it belongs to the Karp's 21 NP-complete problems in~\cite{K72}. A range of optimization methods have been developed to either approximate or solve it exactly including SDR relaxation~\cite{GW95} and the sum-of-squares relaxation in~\cite{KB21,MRX20}.}
Following the idea of Goemans-Williamson relaxation of max-cut~\cite{GW95}, we have an semidefinite relaxation (SDR) of~\eqref{def:ls}. This relaxes $\BR\BR^{\top}\in\RR^{nd\times nd}$ to a set that is positive semidefinite with its diagonal blocks equal to $\I_d$, 
\begin{equation}\label{def:sdr}
\min_{\BX\in \RR^{nd\times nd}}~-\lag \BA,\BX\rag~~~\text{s.t.}~~\BX\succeq 0,~\BX_{ii} = \I_d,
\end{equation}
{where $\BX\succeq 0$ means $\BX$ is positive semidefinite and $\BX_{ii} = \I_d$ assumes the $i$-th $d\times d$ diagonal block is $\I_d$.}
This convex relaxation can be solved by algorithms such as interior point method~\cite{BN01,NN94}.
Its performance has been studied for the synchronization of $\mathbb{Z}_2$-group~\cite{ABBS14}, circle group on torus~\cite{BBS17,S11}, orthogonal group $\Od(d)$~\cite{B15,LYS23,MMMO17}, and special Euclidean group~\cite{RCBL19}. The performance of SDP relaxation is remarkable: under several statistical models, it has been shown that the tightness holds~{\cite{ABBS14,BBS17,B18,L23,L22,L23b}}, i.e.,~\eqref{def:sdr} retrieves the global minimizer to~\eqref{def:ls}. This means that the global minimizer to~\eqref{def:sdr} is exactly rank-$d$ and thus it also corresponds to the global optimal solution of~\eqref{def:ls}. 
Using the Karush-Kuhn-Tucker (KKT) condition of~\eqref{def:sdr}, it is straightforward to argue the tightness holds if there exists $\widehat{\BO}\in\Od(d)^{\otimes n}$ that is rank-$d$ and satisfies:
\begin{equation}\label{def:cert}
\BL = \BDG(\BA\widehat{\BO}\widehat{\BO}^{\top}) - \BA\succeq 0,~~\BL\widehat{\BO}=0
\end{equation}
where 
\begin{equation}\label{def:bdg}
\BDG(\BX) := \frac{1}{2}\blkdiag\left(\BX_{11}+\BX_{11}^{\top},\cdots,\BX_{nn} +\BX_{nn}^{\top}\right). 
\end{equation}
and $\BX_{ii}$ is the $i$-th $d\times d$ diagonal block of $\BX\in\RR^{nd\times nd}.$ {The interested readers may refer to~\cite{B15,BVB20,L23} for the detailed derivation of~\eqref{def:cert}.}
We will often refer to $\BL$ as the certificate matrix or Laplacian matrix since~\eqref{def:cert} is similar to the graph Laplacian.

One benchmark statistical model is the $\Od(d)$-synchronization under additive Gaussian noise~\cite{BBS17,L23,MMMO17,S11}:
\begin{equation}\label{eq:odsync}
\BA_{ij} = \BO_i\BO_j^{\top} +\sigma\BW_{ij} \Longleftrightarrow \BA = \BO\BO^{\top} +\sigma\BW
\end{equation}
where $\BW$ represents an $nd\times nd$ symmetric Gaussian random matrix. 
For $d=1$ which corresponds to $\mathbb{Z}_2$-synchronization, a series of works have investigated the tightness of SDR~\cite{ABBS14,ABH15,B18} and showed the tightness holds if $\sigma\lesssim\sqrt{n/\log n}$. 
However, for $d\geq 2$, it is non-trivial to show the existence of an element $\widehat{\BO}\in\Od(d)^{\otimes n}$ that satisfies~\eqref{def:cert}, mainly because the ground truth $\BO\BO^{\top}$ is not usually the global minimizer to~\eqref{def:sdr}. The first tightness result of the SDR for $d\geq 2$ is derived in~\cite{BBS17} which considered the angular synchronization (it is equivalent to SO($2$)-synchronization) and proved the tightness holds if $\sigma \lesssim n^{1/4}$. 
Starting from~\cite{BBS17}, the condition on $\sigma$ to guarantee the tightness has been improved over the past few years~\cite{B16,LYS17,ZB18}. Using the leave-one-out technique,~\cite{ZB18} provided a near-optimal guarantee on tightness under $\sigma\lesssim \sqrt{n/\log n}$, and later~\cite{L22} extended the result to $\Od(d)$-synchronization provided that $\sigma\lesssim \sqrt{n}/d$ holds (modulo constant factors and logarithmic terms).

\paragraph{Low-rank optimization via Burer-Monteiro factorization}
However, given the high computational costs of solving~\eqref{def:sdr}, one often wants to leverage the low-rank structure of the solution to design more efficient algorithms. A common approach is the Burer-Monteiro factorization~\cite{BM03,BM05}, which ``interpolates" between~\eqref{def:ls} and the SDR by relaxing an orthogonal matrix to an element on Stiefel manifold $\St(p,d)$ defined as
\[
\St(p,d)  : = \{\BS_i\in\RR^{d\times p} : \BS_i\BS_i^{\top} = \I_d\},~~~p\geq d.
\]
Then the Burer-Monteiro factorization of~\eqref{def:sdr} becomes
\begin{equation}\label{def:bm}
\min_{\BS_i\in\St(p,d)}~f(\BS) :=-\sum_{i<j} \lag \BA_{ij}, \BS_i\BS_j^{\top}\rag = -\frac{1}{2}\lag \BA, \BS\BS^{\top}\rag
\end{equation}
where $\BS^{\top} = [\BS_1^{\top},\cdots,\BS_n^{\top}]\in\RR^{p\times nd}.$
In particular, if $p = d$, $\St(p,d)$ is exactly equal to orthogonal group $\Od(d)$, and if $p$ is sufficiently large,~\eqref{def:bm} is equivalent to the convex relaxation~\eqref{def:sdr}. 
%In particular, for $d=1$, $\St(p,d)$ equals the unit sphere $\mathbb{S}^{p-1}$ in $\RR^p.$ 
For the algorithms that solve the problem~\eqref{def:bm}, one may find the references such as~\cite{AMS08,B23,WY13} useful. 

Note that the objective function~\eqref{def:bm} is nonconvex and therefore, first-order methods might get stuck at those spurious local minimizers that are not global. Those spurious local minimizers are more likely to appear for smaller $p$. 
However, in the examples such as $\mathbb{Z}_2$- and $\Od(d)$-synchronization, the algorithms converge to the global minimizer   even with~\emph{random initializations} and very small $p$. This leads to an important theoretical question on the optimization landscape of~\eqref{def:bm}: 
\begin{equation*}
\text{\em Is the optimization landscape benign?}
\end{equation*}
Here we say the optimization landscape is {\em benign} if there is only one local minimizer that is also global in~\eqref{def:bm}. In other words, the objective does not have any spurious local minimizers that are non-global. The characterization of the landscape depends on $p$, $d$, as well as the input data matrix $\BA$. We will now briefly review the existing results on the optimization landscape, and discuss its connection to high-dimensional Kuramoto model.

\paragraph{Optimization landscape}
The optimization landscape analyses of nonconvex functions have been {carried} out in many signal processing and machine learning problems such as~\cite{GLM16,SQW16,SQW18}. For the Burer-Monteiro factorization of the form~\eqref{def:bm}, the work~\cite{BVB20} by Boumal, Voroninski, and Bandeira has proven that a benign landscape holds if $p(p+1)/2 > m$ where $m$ is the number of linear equality constraints. In addition, this bound has been shown to be optimal in general~\cite{P98,WW20}. On the other hand, the required degree of freedom $p$ is notably small if the input data are drawn from certain generative models. For instance, for $\mathbb{Z}_2$-synchronization, i.e., $d=1$ in~\eqref{eq:odsync}, it has been shown in~\cite{BBV16} that as long as $\sigma \lesssim O(n^{1/6})$, then $p=2$ is sufficient to guarantee a benign landscape. Similar results also hold for the community detection under the stochastic block models. For $d\geq 2$,~\cite{L23} provided a general sufficient condition that guarantees the absence of spurious local minimizers in the optimization landscape for $p\geq 2d+1$. In other words, for the practical problem with $p$ moderately larger than $d$, solving~\eqref{def:bm} by gradient descent on the manifold leads to a globally optimal solution, without getting stuck at those local minima~\cite{BBV16,L23}. 
This observation inspires a conjecture that the optimization landscape is benign:~\eqref{def:bm} has only one local minimizer which is also global, if the signal-to-ratio (SNR) is sufficiently large. 

Despite the tightness of the SDR holds under near-optimal condition on the SNR~\cite{B18,L22,ZB18}, it remains an open problem to understand whether the optimization landscape is benign for the SNR in a near-optimal regime. Consider $\mathbb{Z}_2$ as an example:{~\cite{B18} has shown the SDR recovers the planted ground truth if $\sigma \leq\sqrt{\frac{n}{(2+\eps)\log n}}$ for some $\eps>0$ which is information-theoretically optimal~\cite{B18,PK24}}. This threshold is significantly larger than the current theoretical bound on $\sigma\lesssim n^{1/4}$ in~\cite{L23} to guarantee a benign landscape.
This leads one key question of our interests in this work: 
\begin{align*}
\text{\em Is the optimization landscape benign under nearly identical conditions as the SDR succeeds?} 
\end{align*}
We will provide a positive answer: it applies to $\mathbb{Z}_2$-synchronization, community detection under stochastic block model, $ \Od(d)$-synchronization, and generalized orthogonal Procrustes problem.

\paragraph{High-dimensional Kuramoto model}
The problem~\eqref{def:bm} has attracted much attention in the recent years due to its close relation with the high-dimensional Kuramoto model~\cite{K75,MTG20}. More precisely, if we assume $\BA_{ij} = a_{ij}\I_d$, then the resulting objective function is
\begin{equation}\label{def:hikura}
\min_{\BS_i\in \St(p,d)}~f(\BS) := -\sum_{i<j}a_{ij} \lag \BS_i,\BS_j\rag
\end{equation}
which is equivalent to
\begin{equation}\label{def:hikura2}
\min_{\BS_i\in \St(p,d)}~\sum_{i<j}a_{ij} \| \BS_i - \BS_j\|_F^2.
\end{equation}
We call the matrix $\BA$ coupling matrix: $i$-th and $j$-th oscillators are attractive (or repulsive) if $a_{ij}> 0$ (or $a_{ij}<0$).
The high-dimensional homogeneous Kuramoto model on Stiefel manifold is derived by considering the gradient flow of $f(\BS)$ on the manifold:
\begin{equation}\label{def:gs}
\frac{\diff \BS_i(t)}{\diff t} = -\PP_{\BS_i}\left( \frac{\pa f}{\pa \BS_i} \right)
\end{equation}
where $\PP_{\BS_i}(\cdot)$ denotes the projection of the input onto the tangent space of $\BS_i$ on the ${\cal M} := \St(p,d)$:
\begin{equation}\label{def:tan}
T_{\BS_i}({\cal M}) = \{\BPhi: \BPhi\BS_i^{\top} + \BS_i\BPhi^{\top} = 0\}.
\end{equation}
In particular, if $d=1$ and $p=2$, each $\BS_i$ can be parameterized as $\BS_i = (\cos\theta_i,\sin\theta_i)\in\RR^2$, then the gradient system simplifies to 
\[
\frac{\diff \theta_i}{\diff t} = \sum_{j=1}^n a_{ij}\sin(\theta_j-\theta_i)
\]
that exactly equals the celebrated homogeneous Kuramoto model~\cite{K75}. 

The core question regarding the gradient system~\eqref{def:gs} is whether all the oscillators will finally synchronize {globally}: $\lim_{t\rightarrow\infty}\| \BS_i(t)- \BS_j(t)\|_F = 0$ for all $i\neq j$ {from a generic initialization}. Suppose that the Laplacian associated with the coupling matrix $\{a_{ij}\}_{i,j}$ is positive semidefinite, then the fully synchronized state $\BS_i = \BS_j$ achieves the global minimizer to the energy function~\eqref{def:hikura2}, which follows from the SDR~\eqref{def:sdr} and its optimality condition~\eqref{def:cert}.
However, it is not clear in general whether the system exhibits multi-stability, i.e., the system~\eqref{def:gs} has multiple stable equilibria besides the fully synchronized state. {Therefore, the landscape analysis of~\eqref{def:hikura2} is a crucial step towards a comprehensive understanding of the global synchronization of high-dimensional Kuramoto model from any generic initialization.}

A very recent work~\cite{MB23} suggests as long as $p \geq d+2$ and the network is connected with no repulsion terms (none of the $a_{ij}$ are negative), the landscape is benign and the system admits global synchrony. This extends and refines the results in~\cite{MTG17,MTG20}. 
The scenario with $d=1$ and $p=2$ is particularly interesting: several works have shown that the optimization landscape is benign for sufficiently dense networks~\cite{LXB19,KST21,TSS20}. Another compelling instance involves the global synchronization on random graphs: it is shown to achieve global synchronization on Erd\"os-R\'enyi graph~\cite{ABK22,KST22} with high probability provided that the random graph is connected. However, our understanding of the global synchronization for {\em signed} networks is considerably less developed. Therefore, this work aims to answer the following question:
\begin{equation*}
\text{\em Does the high-dimensional Kuramoto model synchronize globally with many repulsive edges?}
\end{equation*}
Our answer is positive under a statistical model: roughly speaking, even if nearly (but fewer than) half of edges are repulsive $(a_{ij} < 0)$, the high-dimensional Kuramoto model still has the global synchrony.

\paragraph{Monotone adversaries}

Spectral methods~\cite{A17,AFWZ20} and SDR~\cite{ABH15,B18,HWX16} have been successfully employed to recover the planted partition in the community detection under the stochastic block model (SBM). There is interest in designing algorithms to solve combinatorial problem such as matrix completion~\cite{CG18} and community detection~\cite{MPW16} that are robust to semi-random model~\cite{FK01}. One such semi-random model is termed ``monotone adversaries" which involves modifying the original data in a way that seemingly helps the recovery. For example, for an adjacency matrix drawn from the SBM, we may add edges between vertices in the same community and delete some edges between different communities. However, it is not hard to show spectral methods or other local methods fail to handle the monotone adversaries in community detection~\cite{MPW16} and planted clique problem~\cite{WX23}. On the other hand, the SDR has been shown to be robust against such monotone adversaries~\cite{FK01}. Given that the SDR is expensive to solve, one is interested in finding an efficient algorithm that remains robust to monotone adversaries and delivers comparable results. The Burer-Monteiro factorization could serve as a potential candidate for this purpose. However, we know that this factorization approach is not completely robust for $d=1$ and $p=2$. In~\cite{LXB19}, it is found that adding one more sample may introduce a new local  (but non-global) minimizer in the optimization landscape of the objective function. Our focus here is to answer:
\begin{equation*}
\text{\em Is the Burer-Monteiro factorization robust to monotone adversaries if } p > 2?
\end{equation*}
We provide a solution to the case when $p\geq 4$, and prove that adding ``favorable adversaries" will not worsen the optimization landscape, i.e., no new local minimizers are created.

In summary, our work makes several key contributions. Firstly, we provide a deterministic guarantee of a benign optimization landscape, which only depends on $p$, $d$, and the condition number of the certificate matrix (also referred to as the Laplacian)~\eqref{def:cert} at the global minimizer. 
In other words, {\em the local geometry at the global minimizer determines the global optimization landscape}. {To the best of our knowledge, it is the first result of this kind that guarantees the global landscape with one key parameter, i.e., the condition number of the Laplacian. Our results establish a general theorem that characterizes the global optimization landscape of the Burer-Monteiro factorization. In other words, if one has established the tightness of the SDR for synchronization, one can immediately show the benignness of the optimization landscape w.r.t. the Burer-Monteiro factorization by using the main results in this work.}

Given the generality of our main theorem, it easily applies to high-dimensional Kuramoto model, $\mathbb{Z}_2$-synchronization and the community detection under the SBM. In these contexts, we demonstrate that for $d=1$ and $p\geq 4$, the optimization landscape of~\eqref{def:bm} is benign even if the signal-to-noise ratio (SNR) in the data is close to the performance bound offered by the SDR. {This significantly improves the existing results~\cite{BBV16,L23,MMMO17} in the aforementioned examples and reduces the general Pataki's bound~\cite{BM05,P98,BVB20} $p(p+1) > n$ to $p\geq 4$ for a large family of models.}
Moreover, we show that the Burer-Monteiro factorization is robust against monotone adversaries: introducing ``favorable" modifications to the data will not worsen the landscape. {This robustness result applies to several semi-random models.}
Lastly, we extend these findings to the case of $d\geq 2$, providing a guarantee on the optimization landscape arising in orthogonal group synchronization and generalized Procrustes problem as special cases.  {Our findings extend the benign landscape analysis of high-dimensional Kuramoto model on sphere and Stiefel manifold~\cite{MTG17,MB23,MTG20} from nonnegative adjacency matrix to any general coupling matrix whose corresponding ``certificate" matrix is well-conditioned.}
%Our result extends and improves the previous work in several ways:
%\item Our result improves the previous work on the benign landscape analysis of the Burer-Monteiro factorization for $\mathbb{Z}_2$-synchronization, community detection, $\Od$-synchronization, and Procrustes problem~\cite{BBV16,L23,MMMO17}. The current sufficient condition on the benign landscape almost matches the information-theoretical limits for the exact recovery by the SDR~\cite{B18}. 

 %Our result improves the general characterization of the optimization landscape of Burer-Monteiro factorization by reducing the Pataki's bound~\cite{BM05,P98,BVB20} $p(p+1) > n$ to $p\geq 4$ for a large family of models and applications by using the condition number of the certificate matrix. 

%Our findings extend~\cite{MB23} to signed networks and more importantly, our results are also applicable to very general coupling matrices. 

\section{Main theorems}

Before proceeding, we go over some notations that will be used. We denote boldface $\bx$ and $\BX$ as a vector and matrix respectively, and $\bx^{\top}$ and $\BX^{\top}$ are their corresponding transpose. The $\I_n$ and $\BJ_n$ stands the $n\times n$ identity  and constant ``1" matrix of size $n\times n$. For a vector $\bx$, $\diag(\bx)$ is a diagonal matrix whose diagonal entries are given by $\bx.$
For two matrices $\BX$ and $\BY$ of the same size, $\BX\circ \BY$ is the Hadamard product, $\lag \BX,\BY\rag = \sum_{i,j}X_{ij}Y_{ij}$ is their inner product, and $\BX\otimes\BY$ is their Kronecker product. For any matrix $\BX$, the Frobenius, nuclear, and operator norms are denoted by $\|\BX\|_F$, $\|\BX\|_*$, and $\|\BX\|$ respectively. For two nonnegative functions $f(x)$ and $g(x)$, we say $f(x)\lesssim g(x)$ (and $f(x) \gtrsim g(x)$) if there exists a universal constant $C> 0$ such that $f(x) \leq Cg(x)$ (and $f(x)\geq Cg(x)$).

\vskip0.2cm

Our main results center on the sufficient condition of a benign optimization landscape of~\eqref{def:bm}. We will divide the discussion into two scenarios: $d=1$ and $d\geq 2$.

\subsection{Theorem for $d=1$}
For $d=1$, the measurement $\BA_{ij} = a_{ij}$ is a scalar, and assume the SDR~\eqref{def:sdr} is tight, i.e., the global minimizer to~\eqref{def:sdr} equals $\BX=\bx\bx^{\top}$ for some $\bx\in\{\pm 1\}^n,$ and 
\begin{equation}\label{def:Ld1}
\BL = \BDG(\BA\bx\bx^{\top}) - \BA\succeq 0
\end{equation}
has the second smallest eigenvalue of $\BL$ strictly positive. It is straightforward to verify that $\BL \bx = 0$ and $\bx$ is the unique global minimizer to~\eqref{def:ls} under~\eqref{def:Ld1}.
We consider the corresponding optimization landscape of~\eqref{def:bm} and our following result gives a sufficient condition for the absence of spurious local minima.

\begin{theorem}\label{thm:main1}
Suppose $\BX=\bx\bx^{\top}$ with $\bx\in\{\pm 1\}^n$ is a global minimizer to~\eqref{def:sdr} and the certificate matrix $\BL$ has a strictly positive second smallest eigenvalue $\lambda_2(\BL)$, {i.e.,~\eqref{def:Ld1} holds}. 
The optimization landscape of~\eqref{def:bm} is benign if
\begin{equation}\label{eq:thm1}
p \geq \frac{2\lambda_{\max}(\BL)}{\lambda_2(\BL)} + 1.
\end{equation}
The only local minimizer {to~\eqref{def:bm}} is also global and satisfies $\BS\BS^{\top} = \bx\bx^{\top}$ with $\BS\in(\mathbb{S}^{p-1})^{\otimes n}$ {where $\mathbb{S}^{p-1}$ is the unit sphere in $\RR^p$ and each row of $\BS$ is in $\mathbb{S}^{p-1}$}.
In particular, every point except the global minimizer has a negative curvature direction. 
\end{theorem}

{We give a few remarks on our result.
Note that this sufficient condition is deterministic and it only depends on the spectra of the certificate matrix $\BL$. It corresponds to the Riemannian Hessian at the global minimizer. Therefore, the condition number of $\BL$, i.e., $\lambda_{\max}(\BL)/\lambda_2(\BL)$ provides a geometric interpretation: if there exists a non-global second order critical point, then the local ``contour" at the global minimizer has to be sufficiently ill-conditioned, i.e., the condition number $\lambda_{\max}(\BL)/\lambda_2(\BL)$ is large enough.

Besides that, Theorem~\ref{thm:main1} implies that every point has a negative curvature direction, i.e., the second derivative along that direction is negative, which means any non-optimal point has at least one direction of descent. Therefore, any algorithm that can take advantage of this property will likely converge to the global minimizer.}
 We are particularly interested in a special family of $\BL$ whose eigenvalues are close to one another, with applications to high-dimensional Kuramoto model, $\mathbb{Z}_2$-synchronization, community detection, {as well as the synchronization problem with incomplete observations.}

\begin{corollary}\label{cor:main1}
Suppose $\BA$ and $\bx$ satisfy~\eqref{def:Ld1}, and moreover
\begin{equation}\label{cond:exp}
\|\BL - r(\I_n - n^{-1}\bx\bx^{\top})\| \leq \alpha r
\end{equation}
for some $0<\alpha<1$ and $r>0$.  Then the optimization landscape is benign if
\[
p \geq \frac{3+\alpha}{1-\alpha}~  \Longleftrightarrow ~0\leq \alpha  \leq \frac{p-3}{p+1}.
\]
\end{corollary}
The corollary above implies that at least we need to have $p\geq 4$ under the assumption~\eqref{cond:exp}. The proof follows simply from~\eqref{cond:exp}, which indicates
\[
(1-\alpha)r \leq \lambda_2(\BL) \leq \lambda_{\max}(\BL) \leq (1+\alpha)r \Longrightarrow \frac{\lambda_{\max}(\BL)}{\lambda_2(\BL)} \leq \frac{1+\alpha}{1-\alpha}
\]
and then applying~\eqref{eq:thm1} in Theorem~\ref{thm:main1} gives the corollary.

Now we proceed to discuss what Corollary~\ref{cor:main1} achieves in applications, and then compare the results with the state-of-the-art sufficient conditions on the benign landscape of $\mathbb{Z}_2$-synchronization and community detection~\cite{BBV16,L23}.

\paragraph{Example: high-dimensional Kuramoto model} 
Given an adjacency matrix $\BA$ with nonnegative entries,~\cite{MB23} implies that for $p\geq 3$ and $d=1$, the corresponding gradient system~\eqref{def:hikura} synchronizes globally. This means there is only one stable equilibrium (all oscillators $\{\BS_i\}_{i=1}^n$ are the same) which is also the unique local minimizer to the energy function~\eqref{def:bm}. However, it is not clear if the coupling $a_{ij}$ between $i$-th and $j$-th oscillators can also be negative, i.e., some pairs of oscillators are repulsive.

We consider a statistical model of signed network and study the global synchrony of the associated high-dimensional Kuramoto model:
\begin{equation}\label{eq:kura_model}
A_{ij} = 
\begin{cases}
1, & \text{with probability }1-\theta, \\
-1, & \text{with probability } \theta.
\end{cases}
\end{equation}
To better represent the model, we let $\BB$ be an adjacency matrix drawn from the Erd\"os-R\'enyi graph with $B_{ij} = B_{ji}\sim$Ber($\theta$), and then
\[
\BA = \BJ_n - 2\BB.
\]
The parameter $\theta$ controls the number of negative edges in the signed network. In particular, if $\theta=0$, i.e., $\BA$ is a complete graph, we know that the corresponding Kuramoto model is globally synchronizing. Now what if the number of negative edges increases, i.e., there are many repulsions? Does the global synchrony still hold? 

We first look into when the global minimum of~\eqref{def:bm} with $\BA = \BJ_n - 2\BB$ is attained at $\BS \BS^{\top} = \BJ_n$. This corresponds to the case in which the fully synchronized states $\BS_i = \BS_j$ achieves the global minimizer to the energy function~\eqref{def:hikura2}.
A sufficient condition is given by
\[
\BL = \BDG(\BA \BJ_n) - \BA \succeq 0
\]
with the second smallest eigenvalue strictly positive, which guarantees the global optimality of $\BJ_n$ to~\eqref{def:sdr}. 

For the Laplacian $\BL$, we have
\begin{align*}
\BL & = \BDG(\E \BA \BJ_n) - \E \BA + \BDG((\BA-\E \BA) \BJ_n) - (\BA-\E \BA) \\
& = (1-2\theta) (n \I_n - \BJ_n) -2\left[ (\BDG(\BB \BJ_n) - \BB) -  (\BDG(\E \BB \BJ_n) - \E \BB) \right] \\
& = (1-2\theta) (n \I_n - \BJ_n) -2\left[ (\BDG(\BB \BJ_n) - \BB) -  \theta n(\I_n - \BJ_n/n) \right] 
\end{align*}
where $\E \BA = (1-2\theta)\BJ_n$, $\E \BB = \theta\BJ_n$ and $\BA - \E \BA = -2(\BB - \E \BB).$ Note that
\[
\BDG(\BB\BJ_n) - \BB = {\underbrace{\sum_{i<j}B_{ij}(\be_i-\be_j)(\be_i-\be_j)^{\top}}_{\BX_{ij}\succeq 0}}
\]
is positive semidefinite where $\be_i$ is the $i$-th one-hot vector{, and it is a sum of $n(n-1)/2$ independent rank-1 positive semidefinite matrices $\{\BX_{ij}\}_{i<j}$.}

Using the Chernoff inequality~\cite[Theorem 1.1]{T12} {with each independent component $\BX_{ij}$ satisfying
\[
\|\BX_{ij}\| = \left\| B_{ij}(\be_i-\be_j)(\be_i-\be_j)^{\top}\right\| \leq 2,
\] 
} the operator norm of $\BDG((\BB-\E \BB) \BJ_n) - (\BB-\E \BB)$ is bounded by
\[
\Pr(\|\BDG(\BB \BJ_n) - \BB- \theta n(\I_n - \BJ_n/n)\| \leq t \theta n) \leq 2n \exp(-t^2 \theta n /6 ) = 2n^{-\gamma}
\]
where {$\sum_{i<j}\E \BX_{ij} = \theta\sum_{i<j}(\be_i-\be_j)(\be_i-\be_j)^{\top} = \theta n (\I_n - \BJ_n/n)$} and
\[
t^2\theta n  =6 (\gamma+1)\log n \Longleftrightarrow t = \sqrt{\frac{6(\gamma+1)\log n}{\theta n}}.
\]
As a result, the eigenvalues of $\BL$ satisfy
\[
(1-2\theta)n - 2 \sqrt{6(\gamma+1)\theta n \log n} \leq \lambda_2(\BL) \leq \lambda_{\max}(\BL) \leq(1-2\theta)n + 2 \sqrt{6(\gamma+1)\theta n\log n}.
\]
By setting $r = 1-2\theta$ and
\[
\alpha = \frac{2}{1-2\theta} \sqrt{\frac{6(\gamma+1)\theta \log n}{n}}
\]
in Corollary~\ref{cor:main1}, it suffices to have $0<\alpha \leq (p-3)/(p+1).$
To simplify this sufficient condition, we use $\theta < 1/2$ and then
\[
\frac{2}{1-2\theta} \sqrt{\frac{6(\gamma+1) \log n}{2n}} \leq \frac{p-3}{p+1} \Longleftrightarrow  \theta \leq \frac{1}{2} - \frac{p+1}{p-3} \sqrt{\frac{3(\gamma+1) \log n}{n}}.
\]
{It is worth noting that using other concentration inequalities such as~\cite{T12,BBV23,BV24} yields a sharper bound on controlling the spectra of $\BL$ for small $\theta$. As we focus on the regime of $\theta < 1/2$, the final resulting bound on $\theta$ will not improve.}
Now we write our discussion into a corollary.
\begin{corollary}\label{cor:km}
Consider the high-dimensional Kuramoto model on the signed network~\eqref{eq:kura_model}.
Suppose the sampling probability $\theta$ of a negative edge satisfies
\begin{equation}\label{cor:km_cond}
\theta \leq \frac{1}{2} - \frac{p+1}{p-3} \sqrt{\frac{3(\gamma+1) \log n}{n}},~~~~p\geq 4,
\end{equation}
then with probability at least $1 - 2n^{-\gamma}$, the high-dimensional Kuramoto model~\eqref{def:gs} has only one stable equilibrium that is also global and equals the fully synchronized state, i.e., $\BS\BS^{\top} = \BJ_n.$
\end{corollary}

In other words, the Kuramoto model~\eqref{def:gs} on $\mathbb{S}^{p-1}$ has the global synchrony even if nearly half of the edges are repulsive for $p\geq 4.$

\paragraph{Example: $\mathbb{Z}_2$-synchronization.} In the $\mathbb{Z}_2$-synchronization under the additive Gaussian noise, the observation equals
\begin{equation}\label{def:z2}
\BA = \bx\bx^{\top} +\sigma\BW
\end{equation}
where $\bx\in\{\pm 1\}^n$ and $\BW$ is an $n\times n$ Gaussian random matrix, i.e., $W_{ij} = W_{ji}\overset{i.i.d.}{\sim}\mathcal{N}(0,1)$ and $W_{ii} = 0$. The aim is to estimate $\bx$ from $\BA$ and the exact recovery via SDR is guaranteed if 
\[
\BL = \BDG(\BA\bx\bx^{\top}) - \BA = n\left(\I_n - \frac{\bx\bx^{\top}}{n}\right) + \sigma(\BDG(\BW\bx\bx^{\top}) - \BW)\succeq 0.
\]
To use Theorem~\ref{thm:main1}, we need to estimate the top and second smallest eigenvalues of $\BL$. % its second smallest value is lower bounded by
%\[
%\lambda_2(\BL) \geq n - \sigma (\|\BW\bx\|_{\infty} + \|\BW\|) \geq n - 4\sigma \sqrt{n\log n}
%\]
%and the largest one satisfies
%\[
%\lambda_{\max}(\BL) \leq n + \sigma (\|\BW\bx\|_{\infty} + \|\BW\|) \leq n + 4\sigma \sqrt{n\log n}
%\]
%where $\|\BW\bx\|_{\infty} + \|\BW\| \leq 4\sqrt{n\log n}$ holds with high probability, following from the fact that each entry is $[\BW\bx]_i$ is $\mathcal{N}(0,n-1)$ and $\|\BW\| \leq 2\sqrt{n}(1+o(1)).$ Based on the calculation above, we can see if $\sigma < 4^{-1}\sqrt{n/\log n}$, then the SDR is tight and recovers $\bx\bx^{\top}$ as the global minimizer. Regarding the optimization landscape of~\eqref{def:bm}, 
%using $\alpha = 4\sqrt{n^{-1}\log n}$ and $r=n$ in Corollary~\ref{cor:main1} leads to the following corollary.
{
It suffices to consider 
\begin{align*}
 \BDG(\BW\bx\bx^{\top}) - \BW & = \sum_{i<j} W_{ij} x_ix_j(x_i \be_i - x_j \be_j)(x_i \be_i - x_j \be_j)^{\top}% \\
% & = \sum_{i<j} W_{ij} x_ix_j(\be_i\be_i^{\top} + \be_j \be_j^{\top} - x_ix_j \be_i\be_j^{\top} - x_ix_j\be_j\be_i^{\top}) \\
% & = \sum_{i\neq j}W_{ij}x_ix_j\be_i\be_i^{\top} - \sum_{i<j} W_{ij} (\be_i\be_j^{\top} + \be_j\be_i^{\top})
 \end{align*}
where $W_{ij}x_ix_j\sim\mathcal{N}(0,1)$. By using~\cite[Theorem 1.2]{T12}, we have
\[
\left\| \sum_{i<j} \left((x_i \be_i - x_j \be_j)(x_i \be_i - x_j \be_j)^{\top}\right)^2 \right\| = 2 \left\| \sum_{i<j} (x_i \be_i - x_j \be_j)(x_i \be_i - x_j \be_j)^{\top} \right\| = 2n
\]
and then
\[
\Pr( \| \BDG(\BW\bx\bx^{\top}) - \BW \|  \geq t) \leq 2n \exp(-t^2/4n) = \exp(\log 2n - t^2/4n).
\]
By picking $t^2 = 4n (1+\eps)\log 2n$, it holds that $\| \BDG(\BW\bx\bx^{\top}) - \BW \| \leq 2\sqrt{n(1+\eps)\log 2n}$ with probability at least $1 - (2n)^{-\eps}$.
 Then using Corollary~\ref{cor:main1} with $r = n$ leads to the following result:
\[
\alpha = \frac{\sigma t}{n} = 2\sigma\sqrt{\frac{(1+\eps)\log 2n}{n}} \leq \frac{p-3}{p+1} \Longleftrightarrow \sigma \leq \frac{p-3}{2(p+1)} \sqrt{\frac{n}{(1+\eps) \log 2n}}
\]
with probability at least $1-(2n)^{-\eps}$ for some $\eps>0.$
}{
\begin{corollary}
Consider the $\mathbb{Z}_2$-synchronization with additive Gaussian noise~\eqref{def:z2}. Suppose
\begin{equation}\label{eq:z2}
2\sigma\sqrt{\frac{(1+\eps)\log 2n}{n}} \leq \frac{p-3}{p+1} \Longleftrightarrow \sigma \leq \frac{p-3}{2(p+1)}\sqrt{\frac{n}{(1+\eps)\log 2n}},
\end{equation}
the optimization landscape of the Burer-Monteiro factorization~\eqref{def:bm} is benign  with probability at least $1-(2n)^{-\eps}$, i.e., the Burer-Monteiro factorization with $p$ satisfying~\eqref{eq:z2} does not suffer from any spurious local minimizer. The only local minimizer $\BS\in(\mathbb{S}^{p-1})^{\otimes n}\subset \RR^{n\times p}$ is global and satisfies $\BS\BS^{\top} = \bx\bx^{\top}.$
\end{corollary}}
The studies on the landscape of~\eqref{def:bm} for the $\mathbb{Z}_2$-synchronization started from~\cite{BBV16} which considered $p=2$ and $d=1$. Theorem 4 in~\cite{BBV16} ensures a benign optimization landscape provided that $\sigma \lesssim n^{1/6}$. This is later improved to $\sigma\lesssim \sqrt{\frac{p-2}{p+1}}n^{1/4}$ in~\cite[Theorem 4]{L23} for $p\geq 3$. Compared with the sufficient condition, i.e., $\sigma\lesssim \sqrt{n/\log n}$ for the exact recovery of $\bx$ via the SDR, both results in~\cite{BBV16,L23} are sub-optimal in terms of $n$. On the other hand, our result shows that the optimization is benign with much larger noise, i.e., $\sigma \lesssim\frac{p-3}{p+1}\sqrt{\frac{n}{\log n}}$ for $p\geq 4$. 
Therefore, ours nearly matches the information-theoretical limit for the exact recovery of $\mathbb{Z}_2$-synchronization up to a constant. 

{Our result has been recently extended to the synchronization with incomplete measurements in~\cite{MABB24}. 
For example, suppose each $A_{ij}$ is observed with probability $\xi$, then a similar argument along with the concentration inequality as shown above can be used to control the top and second smallest eigenvalues of $\BL.$ Using Theorem~\ref{thm:main1} will yield a sufficient condition on $\xi$ and $\sigma$ that ensures the benignness of the optimization landscape.}

\paragraph{Example: community detection under stochastic block model.} The observed adjacency matrix satisfies
\begin{equation}\label{def:sbm}
A_{ij} \sim 
\begin{cases}
\text{Ber}(p_{\rm in})-(p_{\rm in}+p_{\rm out})/2, & (i,j)\text{ in the same community}, \\
\text{Ber}(p_{\rm out})-(p_{\rm in}+p_{\rm out})/2, & (i,j)\text{ in different communities},
\end{cases} 
\end{equation}
where $p_{\rm in} > p_{\rm out}$ and Ber$(p)$ stands for Bernoulli random variable with mean $p$. 
Without loss of generality, we let the first $n/2$ members be the first community and the others be the second one, i.e., the ground truth membership is $\bx = [\bone_{n/2};-\bone_{n/2}]$. The community detection problem asks to recover the hidden membership from the adjacency matrix $\BA.$
The SDR~\eqref{def:sdr} recovers the membership exactly if the ``Laplacian" $\BL$ associated with $\BA$ and $\bx$ has a strictly positive second smallest eigenvalue.

Now we derive a sufficient condition to ensure the benign optimization landscape of the Burer-Monteiro factorization, which reduces to estimate the eigenvalues of $\BL.$
For the adjacency matrix $\BA$ drawn from the SBM, we have
\[
\E \BA =
\begin{bmatrix}
p_{\rm in}\BJ_{n/2} & p_{\rm out}\BJ_{n/2} \\
p_{\rm out}\BJ_{n/2} & p_{\rm in}\BJ_{n/2}
\end{bmatrix} - \frac{p_{\rm in} + p_{\rm out}}{2}\BJ_n = \frac{p_{\rm in} - p_{\rm out}}{2} \bx\bx^{\top}.
\]
As a result, we have
\[
\BL = \underbrace{\frac{p_{\rm in} - p_{\rm out}}{2}(n\I_n - \bx\bx^{\top})}_{\E \BL} + \underbrace{\left( \diag(\bx)\diag( (\BA-\E \BA)\bx ) -(\BA-\E\BA)\right)}_{\BL - \E \BL}.
\]

For the difference $\BL -\E \BL$, using the Bernstein inequality~\cite[Theorem 1.4]{T12} leads to
\[
\|\BL - \E \BL\| \leq C_0 (\sqrt{n(p_{\rm in} + p_{\rm out})\log n} + \log n) \leq C'_0 \sqrt{n(p_{\rm in} + p_{\rm out})\log n}
\]
following from~\cite[Lemma 17 and 18]{BBV16} and $p_{\rm in}\gtrsim n^{-1}\log n.$
Let $r = n(p_{\rm in} - p_{\rm out})/2$ be the operator of $\E\BL$ and then $\alpha$ is bounded by
\[
\alpha \leq \frac{ 2C_0'\sqrt{n(p_{\rm in}+p_{\rm out})\log n}}{n(p_{\rm in}-p_{\rm out})}.
\]
By using Corollary~\ref{cor:main1}, we write the discussion above into a corollary.
\begin{corollary}\label{cor:sbm}
Consider the stochastic block model with two communities~\eqref{def:sbm}. Suppose
\[
\frac{2C_0' \sqrt{n(p_{\rm in}+p_{\rm out})\log n}}{n(p_{\rm in}-p_{\rm out})} \leq \frac{p-3}{p+1},~~~p\geq 4,
\]
then the optimization landscape of the Burer-Monteiro factorization~\eqref{def:bm} is benign with high probability. Every local minimizer $\BS$ satisfies $\BS\BS^{\top} = \bx\bx^{\top}$. In particular, if $p_{\rm in} = an^{-1}\log n$ and $p_{\rm out} = bn^{-1}\log n$, then the condition becomes
\begin{equation}\label{cond:sbm}
\frac{a-b}{ \sqrt{a+b}} \geq 2C_0'\cdot  \frac{p+1}{p-3}.
\end{equation}
\end{corollary}
The exact recovery of community members from the SBM has been studied extensively in~\cite{A17,B18,ABH15,AFWZ20,WZS22}: the exact recovery is possible for spectral methods~\cite{AFWZ20}, convex relaxation~\cite{ABH15,B18}, and projected power method~\cite{WZS22} if and only if $\sqrt{a} - \sqrt{b} > \sqrt{2}$. Our bound~\eqref{cond:sbm} nearly matches the threshold for exact recovery (up to a constant) for $p\geq 4$. On the other hand, for $p=2$, it is shown in~\cite[Theorem 6]{BBV16} that 
\[
\frac{n(p_{\rm in} - p_{\rm out})}{\sqrt{2n(p_{\rm in} + p_{\rm out})}} \gtrsim n^{1/3} \Longleftrightarrow \frac{a-b}{\sqrt{a+b}} \gtrsim \frac{n^{1/3}}{\sqrt{\log n}}
\]
guarantees a benign landscape where the right hand side of the condition depends on $n^{1/3}$. 

In summary, we have provided a sufficient condition for a benign landscape for $p\geq 4$, with the noise level almost matching the information-theoretical bound. 
{After this work is posted, several works have been posted to improve Theorem~\ref{thm:main1}, especially~\cite{EW24,MABB24}. In particular,~\cite{MABB24} improves the bound on $\sigma$ to $\sigma\leq \frac{p-3}{p-1}\sqrt{\frac{n}{(2+\eps)\log n}}$ in the $\mathbb{Z}_2$-synchronization with Gaussian noise. For $p$ sufficiently large, the bound matches the information-theoretical bound in~\cite{B18}. Later,~\cite{EW24} improves $\lambda_{\max}(\BL)/\lambda_2(\BL) \leq (p-1)/2$ in Theorem~\ref{thm:main1} to $\lambda_{\max}(\BL)/\lambda_2(\BL) < p$, which is also proven optimal since a counterexample can be constructed such as a non-global local minimizer exists for some $\BL$ with condition number equal to $p$. In other words, a benign landscape also holds if $p = 2$ and 3 with the noise level almost close to the information-theoretical bound.}
%However, it remains unclear for $p = 2$ and $3$, where the state-of-the-art bound on $\sigma$ (also in the stochastic block model and high-dimensional Kuramoto model) is suboptimal. Therefore, we propose the following conjecture.
%\begin{conjecture}
%The optimization landscape~\eqref{def:bm} is benign for $p=2$ and $3$ with high probability in the following scenarios:
%\begin{itemize}
%\item for the high-dimensional Kuramoto model with the coupling matrix drawn from~\eqref{eq:kura_model}, it suffices to have
%\[
%\frac{1}{2} - \theta  \gtrsim \sqrt{\frac{\log n}{n}},~~~\theta < \frac{1}{2};
%\]

%\item for the $\mathbb{Z}_2$-synchronization under additive Gaussian noise~\eqref{def:z2}, it suffices to have
%\[
%\sigma \lesssim \sqrt{\frac{n}{\log n}};
%\]

%\item for the community detection under the stochastic block model~\eqref{def:sbm},  it suffices to have
%\[
%\frac{a-b}{\sqrt{a+b}} \geq C_0
%\]
%for some constant $C_0$ where $p_{\rm in} = an^{-1}\log n$ and $p_{\rm out} = bn^{-1}\log n$.
%\end{itemize}
%\end{conjecture}

\paragraph{Example: robustness against monotone adversaries}
Now we look into the robustness of the optimization landscape of~\eqref{def:bm} against monotone adversaries: does adding  ``favorable" samples to the original data create new local minimizers? 
We take the SBM as an example, and the argument also applies to other examples.
Suppose $\BA_{M}$ is sampled from the SBM in the form of~\eqref{def:sbm}. Theorem~\ref{thm:main1} implies $p\geq 2\lambda_{\max}(\BL_{M})/\lambda_2(\BL_{M}) + 1$ ensures a benign landscape where $\BL_{M}$ is defined in~\eqref{def:Ld1}. The membership vector $\bx = [\bone_{n/2};-\bone_{n/2}]$ is the ground truth and $\bx\bx^{\top}$ is the solution to the SDR~\eqref{def:sdr}. 

Consider the ``favorable" adversaries added to the SBM, i.e., we add edges between the nodes in the same community and delete those between different communities. The adversary matrix is in the following block structure:
\begin{equation}\label{def:adv}
\BDelta = 
\begin{bmatrix}
+ & - \\
- & +
\end{bmatrix}\in\RR^{n\times n}
\end{equation}
Here $``+"$ ($``-"$) means the entries in the corresponding $n/2\times n/2$ block is non-negative (non-positive). By construction, it satisfies
\[
\BA_{M} + \frac{p_{\rm in} + p_{\rm out}}{2}\BJ_n + \BDelta  \geq 0
\]
as the adversaries only delete existing edges between communities. 
The final data matrix is
\[
\BA :=\BA_{M} + \BDelta
\]
which is used to recover the hidden communities.
The corresponding Laplacian satisfies
\[
\BL = \BL_{M} + \BL_{\Delta} = (\BDG(\BA_{M}\bx\bx^{\top}) - \BA_{M}) + (\BDG(\BDelta\bx\bx^{\top}) - \BDelta)
\]
where $\BL_{M}$ and $\BL_{\Delta}$ represent the ``Laplacian" w.r.t. the original adjacency matrix and the adversary matrix. By construction, $\BL_{M}$ is positive semidefinite and so is $\BL_{\Delta}$ since $\diag(\bx)\BL_{\Delta}\diag(\bx)$ is exactly the graph Laplacian w.r.t. the nonnegative adjacency matrix $\diag(\bx)\BDelta\diag(\bx).$
Therefore, we have
\[
\BL \bx = 0,~~~\BL\succeq 0
\]
which guarantees $\bx\bx^{\top}$ is the global minimizer to the SDP relaxation. It is not clear immediately that $p \geq 2\lambda_{\max}(\BL_{M})/\lambda_2(\BL_{M}) + 1$ still suffices to ensure a benign landscape, because adding $\BL_{\Delta}$ to $\BL_{M}$ potentially changes the condition number.

Our next theorem shows that as long as the optimization landscape associated with the original data is benign, then adding monotone adversaries does not make the landscape worse: no new local minimizers will be created if $p$ satisfies $p \geq 2\lambda_{\max}(\BL_{M})/\lambda_2(\BL_{M}) + 1$. The theorem holds in more general scenarios besides the SBM. This robustness against monotone adversaries does not hold for $p=2$: as shown in~\cite{LXB19}, adding one more sample can lead to a new spurious local minimizer that is non-global.

\begin{theorem}[\bf Robustness against monotone adversaries]\label{thm:adv}
Suppose the optimization landscape of~\eqref{def:bm} is benign for the data matrix $\BA_{M}$, i.e., 
\[
p \geq \frac{2\lambda_{\max}(\BL_{M}) }{\lambda_2(\BL_{M})}+1,~~\BL_M = \BDG(\BA_{M}\bx\bx^{\top}) -\BA_{M}\succeq 0.
\]
Then for any monotone adversaries in the form of~\eqref{def:adv} with
\[
\BDG(\BDelta\bx\bx^{\top}) - \BDelta\succeq 0,~~\BDelta\circ \bx\bx^{\top}\geq 0,
\]
the resulting optimization landscape is still benign for the same $p$. The only local minimizer is given by $\BS\in(\mathbb{S}^{p-1})^{\otimes n}$ satisfying $\BS\BS^{\top} = \bx\bx^{\top}$, which is also global.
\end{theorem}
The assumption of the nonnegativity of $\diag(\bx)\BDelta\diag(\bx)\geq 0$ is crucial, as the proof of Theorem~\ref{thm:adv} essentially combines Theorem~\ref{thm:main1} with the result in~\cite{MB23} which requires the non-negativity. 

Theorem~\ref{thm:adv} can be also used in a different way to achieve a possibly better bound on $p$.  Given any data matrix $\BA$, and assume the SDP relaxation~\eqref{def:sdr} yields a global minimizer $\BX = \bx\bx^{\top}$ and $\BL = \BDG(\BA\bx\bx^{\top}) - \BA\succeq 0$. Now we are interested in how large $p$ is needed to guarantee a benign landscape of~\eqref{def:bm}? A simple answer is to use $p \geq 2\lambda_{\max}(\BL)/\lambda_2(\BL) + 1.$ Theorem~\ref{thm:adv} implies one can decrease $p$ if there exists a decomposition of $\BA$ into $\BA_M$ and $\BDelta$ such that $\BL_M = \BDG(\BA_M \bx\bx^{\top}) - \BA_M$ is positive semidefinite and better conditioned, and $\BDelta\circ \bx\bx^{\top}$ is nonnegative. Then we can simply set $p \geq 2\lambda_{\max}(\BL_M)/\lambda_2(\BL_M) + 1$ if the condition number of $\BL_M$ is smaller than that of $\BL.$ 

{ Based on the discussion above, we introduce another semi-random variant of the standard stochastic block model, inspired by~\cite{CG18}. For any pair of node $(i,j)$ in the same community, an edge $A_{ij}$ is sampled with probability $p_{ij}$ that is at least $p_{\rm in}$; and for any pair node $(i,j)$ in different communities, an edge $A_{ij}$ is sampled with probability $p_{ij}$ that is no greater than $p_{\rm out}.$ In other words, more edges are randomly added within the same communities and fewer edges are sampled between different ones. Now we will decompose $\BA$ into $\BA_M$ and $\BDelta$ such that Theorem~\ref{thm:adv} applies.
Let
\[
M_{ij}\sim 
\begin{cases}
\text{Ber}(p_{\rm in}/ p_{ij}), & (i,j)\text{ in the same community}, \\
\text{Ber}((p_{\rm out} - p_{ij})/(1-p_{ij})), & (i,j)\text{ in different communities},
\end{cases}
\]
be independent of $\BA$
and
\[
\BA_{M} = 
\begin{cases}
A_{ij} M_{ij}, & (i,j)\text{ in the same community}, \\
A_{ij} + (1-A_{ij}) M_{ij}, & (i,j)\text{ in different communities}.
\end{cases}
\]
Note that $\E A_{ij}M_{ij} = p_{ij}\cdot p_{\rm in}/ p_{ij}= p_{\rm in} $ for $(i,j)$ in the same community and 
\[
\E\left[ A_{ij} + (1-A_{ij}) M_{ij} \right]= p_{ij} + (1-p_{ij})\cdot \frac{p_{\rm out} - p_{ij}}{1- p_{ij}} = p_{\rm out}
\]
for $(i,j)$ in different communities.
Here
\[
\BDelta = \BA - \BA_{M} = 
\begin{cases}
A_{ij}(1- M_{ij}), & (i,j)\text{ in the same community}, \\
- (1-A_{ij}) M_{ij}, & (i,j)\text{ in different communities}.
\end{cases}
\]
In other words, conditioned on $\BA$ drawnly from the semi-random SBM, we can randomly delete edges in the same community and add edges between different communities such that $\BA_M$ is exactly sampled from the standard SBM with parameters $(n,p_{\rm in},p_{\rm out})$. The noise matrix $\BDelta$ obviously satisfies $\diag(\bx)\BDelta\diag(\bx)\geq 0$. In this case, as long as the Burer-Monteiro factorization associated with $\BA_M$ has a benign landscape, i.e., $(n,p_{\rm in},p_{\rm out},p)$ satisfies the assumption in Corollary~\ref{cor:sbm}, then so does the one corresponding to $\BA$. We write this discussion into the following corollary.
\begin{corollary}
Consider the semi-random variant of the stochastic block model, i.e., 
\[
A_{ij}\sim 
\begin{cases}
\text{Ber}(p_{ij}), & p_{ij} \geq p_{\rm in},~~(i,j)\text{ in the same community}, \\
\text{Ber}(p_{ij}), & p_{ij} \leq p_{\rm out},~~(i,j)\text{ in different communities}.
\end{cases}
\]
Suppose $p_{\rm in}$, $p_{\rm out}$, and $p$ satisfies the assumption in Corollary~\ref{cor:sbm}, then the optimization landscape of the Burer-Monteiro factorization corresponding to $\BA$ is benign, i.e., the only local minimizer is given by $\BS\in(\mathbb{S}^{p-1})^{\otimes n}$ satisfying $\BS\BS^{\top} = \bx\bx^{\top}$, which is also global.
\end{corollary}
}

\subsection{Theorem for $d\geq 2$}
For $d\geq 2$, the global minimizer of~\eqref{def:sdr} is shown to be rank-$d$ in the high SNR regime and thus equal to that of~\eqref{def:ls} in several applications. However, the solution is typically not equal to the ground truth~\cite{BBS17,ZB18,L22}, as shown in the examples such as angular synchronization~\cite{BBS17,ZB18}, orthogonal group synchronization~\cite{L22}, and generalized orthogonal Procrustes problem~\cite{L23b}.

From now on, we consider the convex relaxation~\eqref{def:sdr}, and assume $\widehat{\BO}$ satisfies~\eqref{def:cert} and $\widehat{\BO}\widehat{\BO}^{\top}$ is the unique global minimizer to~\eqref{def:sdr}. Our focus is on the optimization landscape of~\eqref{def:bm} that is used to solve~\eqref{def:sdr}. 
 
\begin{theorem}\label{thm:main2}
Suppose $\BX = \widehat{\BO}\widehat{\BO}^{\top}$ with $\widehat{\BO}\in\Od(d)^{\otimes n}$, and the certificate matrix $\BL$ satisfies
\[
\BL = \BDG(\BA\widehat{\BO}\widehat{\BO}^{\top}) - \BA, ~~~\BL\widehat{\BO} = 0,~~~\BL\succeq 0,~~~{\lambda_{d+1}(\BL) > 0},
\]
i.e., $\BX$ is the {unique} global minimizer to~\eqref{def:sdr}.
The optimization landscape of the Burer-Monteiro factorization~\eqref{def:bm} is benign if 
\[
p \geq \left(\frac{2\lambda_{\max}(\BL)}{\lambda_{d+1}(\BL)} - 1\right)d + 2.
\]
The only local minimizer $\BS$ to~\eqref{def:bm} is also global and satisfies $\BS\BS^{\top} = \widehat{\BO}\widehat{\BO}^{\top}.$
In particular, every point except the global minimizer has a negative curvature direction. 
\end{theorem}

{We make a comparison between Theorem~\ref{thm:main1} and~\ref{thm:main2}. The result above  generalizes Theorem~\ref{thm:main1} from $d=1$ to $d\geq 2$. For $d=1$, we know that the global minimizer satisfies $\widehat{\BO} \in\{\pm 1\}^n$ if the tightness holds, and thus Theorem~\ref{thm:main2} implies Theorem~\ref{thm:main1}. 
Regarding their assumptions, we note that for $d= 1$, the certificate matrix in~\eqref{def:Ld1} immediately satisfies $\BL\bx = 0$, and $\lambda_2(\BL)>0$ suffices to guarantee the tightness of the SDR. However, for $d\geq 2$, $\lambda_{d+1}(\BL) > 0$ alone does not guarantee the exactness of the SDR and the global optimality of $\widehat{\BO}\widehat{\BO}^{\top}$. In addition, $\BL\widehat{\BO} = 0$ is required in Theorem~\ref{thm:main2} because it does not hold by its definition. In fact, the condition $\BL\widehat{\BO} = 0$ assumes that $\widehat{\BO}$ is a first-order critical point in~\eqref{def:bm}.}

Similar to Theorem~\ref{thm:main1}, Theorem~\ref{thm:main2} is quite general as the condition is purely algebraic: it says that the local geometry at the global minimizer governs the global optimization landscape. To use Theorem~\ref{thm:main2}, it suffices to estimate the top and $(d+1)$-th smallest eigenvalues of $\BL = \BDG(\BA\widehat{\BO}\widehat{\BO}^{\top}) - \BA$. This task is typically harder for $d\geq 2$ than $d=1$ as the global minimizer does not have a closed form solution for $d\geq 2$.
Now we proceed to discuss the two applications of Theorem~\ref{thm:main2} in which the condition number of $\BL$ is well understood.

\paragraph{Example: $\Od(d)$-synchronization under additive Gaussian noise.} 
We consider the orthogonal group synchronization under additive Gaussian noise, as introduced in~\eqref{eq:odsync}. In~\cite{L22}, it is shown that the SDR~\cite{L22} is tight with high probability under $\sigma\lesssim \sqrt{n}/(\sqrt{d}(\sqrt{d} +\sqrt{\log n}))$. In fact,~\cite{L22} proved that the generalized power method (a nonconvex iterative algorithm for $\Od(d)$-synchronization) converges globally to the SDR solution. Despite the analysis requires spectral initialization, the global convergence still holds empirically even if we randomly initialize the algorithm. To understand this, it is natural to look into the optimization landscape of Burer-Monteiro factorization~\eqref{def:bm} for the $\Od(d)$-synchronization. 
As the tightness of the SDR is guaranteed in~\cite{L22}, it suffices to estimate the top and $(d+1)$-th smallest eigenvalues of $\BL = \BDG(\BA\widehat{\BO}\widehat{\BO}^{\top}) - \BA$ via using the by-products in~\cite{L22}, and then apply Theorem~\ref{thm:main2}. The result is summarized as follows, and the proof is deferred to Section~\ref{ss:cor}.

\begin{corollary}[\bf Landscape of orthogonal group synchronization]\label{cor:odsync}

For the orthogonal group synchronization under additive Gaussian noise~\eqref{eq:odsync}, 
the optimization landscape of~\eqref{def:bm} is benign with high probability provided that
\[
\sigma  \leq \frac{p-d-2}{p+3d-2}\cdot\frac{\sqrt{n}}{C \sqrt{d}(\sqrt{d} + 4\sqrt{\log n})}
\]
for some positive constant $C>0.$
\end{corollary}

Corollary~\ref{cor:odsync} implies that if $p\geq 2d+2$, then
\[
\frac{p-d-2}{p+3d-2} \geq \frac{1}{5}
\]
which implies that the optimization landscape is benign with high probability for $\sigma\lesssim \sqrt{n}/(\sqrt{d}(\sqrt{d} + 4\sqrt{\log n}))$. In other words, for $p\geq 2d +2$, the Burer-Monteiro factorization does not suffer from any spurious local minimizers under almost identical assumption that ensures the tightness of the SDR.

\paragraph{Example: Generalized orthogonal Procrustes problem.}

Consider $\BA_i$ is a noisy copy of a point cloud $\bar{\BA}$ transformed by an unknown orthogonal matrix $\BO_i$ and additive noise, i.e.,
\begin{equation}\label{eq:gopp}
\BA_i = \BO_i \bar{\BA} + \sigma \BW_i,~~1\leq i\leq n,
\end{equation}
where $\bar{\BA}\in\RR^{d\times m}$ is a point cloud consisting of $m$ points in $\RR^d$. The generalized orthogonal Procrustes problem is to align $n$ point clouds by estimating $\bar{\BA}$ and $\{\BO_i\}_{i=1}^n$ from $\{\BA_i\}_{i=1}^n.$ It has found many applications in statistics~\cite{G75}, computer vision~\cite{MDKK16,MGPG04}, and imaging science~\cite{BKS14,CKS15,S18,SS11}.

The least squares estimation of $\BO_i$ and $\bar{\BA}$ is equivalent to 
\begin{equation}\label{def:ls_gopp}
\min_{\BR\in\Od(d)^{\otimes n}}~-\sum_{i<j}\lag \BA_i\BA_j^{\top}, \BR_i\BR_j^{\top} \rag.
\end{equation}
Suppose $\widehat{\BO}_i$ is the global minimizer to~\eqref{def:ls_gopp}, then the estimation of $\bar{\BA}$ is given by
$\widehat{\bar{\BA}} = n^{-1} \sum_{i=1}^n \widehat{\BO}_i^{\top}\BA_i$.
 We can see that~\eqref{def:ls_gopp} is identical to~\eqref{def:ls} by setting $\BA_{ij} = \BA_i\BA_j^{\top}$, and it seems NP-hard to find the global minimizer to~\eqref{def:ls_gopp}.
The work~\cite{BKS14} by Bandeira, Khoo, and Singer conjectured that if the noise level $\sigma$ is below certain threshold $\sigma^* > 0$,  then the SDR is tight, i.e., it recovers the global minimizer to~\eqref{def:ls_gopp}. 

The tightness of the SDR is answered in~\cite{L23b}: the generalized power method with a proper initialization converges to the global minimizer to the SDR in the form of $\widehat{\BO}\widehat{\BO}^{\top}$ for some $\widehat{\BO}\in\Od(d)^{\otimes n}$ with probability at least $1-O(n^{-\gamma+2})$ if
\[
\sigma \lesssim \frac{\sigma_{\min}(\bar{\BA})}{\kappa^4} \cdot \frac{\sqrt{n}}{\sqrt{d}(\sqrt{nd} + \sqrt{m} + \sqrt{2\gamma n\log n})},~~\kappa = \frac{\sigma_{\max}(\bar{\BA})}{\sigma_{\min}(\bar{\BA})}.
\]
Moreover, the global minimizer $\widehat{\BO}$ satisfies $\BL = \BDG(\BA\widehat{\BO}\widehat{\BO}^{\top}) - \BA\succeq 0$ and $\BL\widehat{\BO}=0.$
Similar to the orthogonal group synchronization, the convergence analysis of generalized power method requires a spectral initialization. However, the generalized power method is empirically successful with random initialization, which leads to a conjecture on the benign optimization landscape of~\eqref{def:bm}. Our next corollary is a direct consequence of Theorem~\ref{thm:main2} by a careful estimation of the condition number of $\BL$ with the results in~\cite{L23b}. 

\begin{corollary}[\bf Optimization landscape of general orthogonal Procrustes problem]\label{cor:gopp}
Consider the generalized orthogonal Procrustes problem with additive Gaussian noise~\eqref{eq:gopp}, the optimization landscape of~\eqref{def:bm} is benign with probability at least $1- O(n^{-\gamma+2})$ provided that 
\[
\sigma \leq \frac{p + d - 2\kappa^2 d- 2 }{p + d + 2\kappa^2 d - 2} \cdot \frac{\sqrt{n}\|\bar{\BA}\|}{C\kappa^4 \sqrt{d} (\sqrt{nd} + \sqrt{m} + \sqrt{ 2\gamma n\log n})}
\]
for some constant $C>0.$
\end{corollary}

Let $p\geq 2\kappa^2 d+2$, then we have
\[
\frac{p + d - 2\kappa^2 d- 2 }{p + d + 2\kappa^2 d - 2} \geq \frac{1}{4\kappa^2 + 1}.
\]
Therefore, if $\sigma \lesssim \sqrt{n}\|\bar{\BA}\|/ (\kappa^6 \sqrt{d} (\sqrt{nd} + \sqrt{m} + 4\sqrt{ n\log n}))$, then the optimization landscape is benign with high probability. Corollary~\ref{cor:gopp} provides a positive answer to the benign landscape under the noise level $\sigma$ that almost matches the sufficient condition for the tightness of the SDR.

{Finally, we conclude our work by highlighting its main contributions. The major contribution of this work is a deterministic condition that guarantees the benign landscape of the Burer-Monteiro factorization arising from group synchronization. 
Notably, the condition number of the Laplacian matrix at the global minimizer plays a pivotal role in characterizing the global landscape—a surprising finding that has inspired a series of subsequent works~\cite{MABB24,EW24}.
Our results are versatile and apply to various statistical models, achieving near-optimal conditions in a broad range of examples such as $\mathbb{Z}_2$-synchronization, community detection, $\Od(d)$-synchronization, and the generalized Procrustes problem. Moreover, we demonstrate that the Burer-Monteiro factorization is robust to monotone adversaries, particularly in community detection and graph clustering.

Several future directions remain to be explored. While our results affirm the benignness of the optimization landscape, a comprehensive understanding of how optimization algorithms, such as Riemannian gradient descent or the projected power method, converge to the global minimizer from generic random initialization is still largely unknown. An even more intriguing question is whether the phenomenon that "local information determines the global landscape" extends to other low-rank matrix recovery problems, including matrix completion and phase retrieval. If so, what is the appropriate quantity that governs the global optimization landscape? Conversely, if not, what are the counterexamples? We leave these questions for future work.}

\section{Proofs}

We will first provide a proof of Theorem~\ref{thm:main2}; and then proceed to use it to justify Corollary~\ref{cor:odsync} and~\ref{cor:gopp}. In fact, Theorem~\ref{thm:main1} for $d=1$ is a special case of Theorem~\ref{thm:main2}. The idea of proof follows from considering the second order necessary condition of~\eqref{def:bm} and then choosing a particular random direction~\cite{MMMO17,MTG20,MB23}. The goal is to show that the quadratic form of the Riemannian Hessian at every point  is non-positive except at the global minimizer, and this implies the only local minimizer is also global. The key technical part is to control the quadratic form of the Hessian via the eigenvalues of $\BL$ at the global minimizer.

We quickly review the Riemannian gradient and Hessian of~\eqref{def:bm} on the Stiefel manifold. {The derivations are standard and can be found in~\cite{B15,BVB20,B23,L23}.} 
The projection of $\BZ\in\RR^{d\times p}$ onto the tangent space at $\BS_i\in\St(p,d)$ is given by
\begin{align*}
\PP_{\BS_i}(\BZ) & = \BZ - \frac{(\BS_i\BZ^{\top} + \BZ\BS_i^{\top})}{2}\BS_i.
\end{align*}

The Riemannian gradient of $f(\BS)$ is given by
\begin{align*}
\nabla_{\BS_i} f(\BS) & = -\PP_{\BS_i}\left(\sum_{j=1}^n \BA_{ij}\BS_j\right) \\
& = -\sum_{j=1}^n \BA_{ij}\BS_j + \frac{1}{2}\sum_{j=1}^n (\BA_{ij}\BS_j\BS_i^{\top} +\BS_i \BS_j^{\top}\BA_{ji}) \BS_i.
\end{align*}
From now on, we denote $\BLambda_{ii} = \frac{1}{2}\sum_{j=1}^n(\BA_{ij}\BS_j\BS_i^{\top} + \BS_i\BS_j^{\top}\BA_{ji})$, and in particular {if $\BS$ is a critical point, i.e., $\nabla_{\BS_i}f(\BS) = 0$, then $\BLambda_{ii}$ equals
\[
\BLambda_{ii} = \sum_{j=1}^n \BA_{ij}\BS_j\BS_i^{\top}
\]
because $\sum_{j=1}^n \BA_{ij}\BS_j\BS_i^{\top} = \sum_{j=1}^n\BS_i\BS_j^{\top}\BA_{ji}$ holds for any first order critical point $\BS\in\St(p,d)^{\otimes n}$.}

The Riemannian Hessian $\nabla^2 f(\BS)$ of $f(\BS)$ in terms of the quadratic form is
\begin{align*}
\dot{\BS} : \nabla^2 f(\BS) : \dot{\BS} := \lag \BLambda - \BA, \dot{\BS}\dot{\BS}^{\top}\rag
\end{align*}
where $\BLambda = \BDG(\BA\BS\BS^{\top})$ and $\dot{\BS}$ is any element belonging to $T_{\BS_1}({\cal M}) \times \cdots \times T_{\BS_n}({\cal M})$ in $(\RR^{d\times p})^{\otimes n}$.
We say $\BS$ is a second order critical point (SOCP) if $\nabla_{\BS}f(\BS) = 0$ and $\dot{\BS} : \nabla^2 f(\BS) : \dot{\BS}\geq 0$ for any $\dot{\BS}$ in the tangent space. The following lemma characterizes when a SOCP $\BS\in\St(p,d)^{\otimes n}$ is the unique global minimizer to~\eqref{def:ls} and~\eqref{def:sdr}.

\begin{lemma}
Consider the SDP relaxation~\eqref{def:sdr} of~\eqref{def:ls}. A solution $\BS\in\St(p,d)^{\otimes n}$ is the unique global minimizer to~\eqref{def:ls} if
\[
\BLambda = \BDG(\BA\BS\BS^{\top}),~~\BLambda - \BA \succeq 0,~~(\BLambda - \BA)\BS = 0,
\]
and the $(d+1)$-th smallest eigenvalue of $\BLambda - \BA$ is strictly positive. 
\end{lemma}
The lemma above follows from the duality theory in convex optimization {such as~\cite{B15,B23,L23}}, and is very useful in characterizing the tightness and global optimality of a candidate solution.

\subsection{Proof of Theorem~\ref{thm:main1} and~\ref{thm:main2}}

Our proof will rely on another two simple yet important lemmas.

\begin{lemma}\label{lem:Xineq}
For any positive semidefinite matrix $\BX\in\RR^{n\times n}$ with diagonal entries equal to $d$, then it holds that
\[
0\leq nd^2 - n^{-1} \|\BX\|_F^2 \leq 2d(nd - n^{-1}\lag \BX,\BJ_n\rag).
\]
\end{lemma}
\begin{proof}
Note that $\lag \BX,\BJ_n\rag \leq n\|\BX\|_F \leq n^2d$ follows from Cauchy-Schwarz inequality. Then 
\begin{align*}
\frac{nd^2 - n^{-1} \|\BX\|_F^2}{nd - n^{-1}\lag \BX,\BJ_n\rag} & =  \frac{n^2d^2 - \|\BX\|_F^2}{n^2d - \lag \BX,\BJ_n\rag} \leq \frac{n^2d^2 - n^{-2}\lag \BX,\BJ_n\rag^2}{n^2d - \lag \BX,\BJ_n\rag} \\
& = \frac{n^2d + \lag \BX,\BJ_n\rag}{n^2} \leq 2d
\end{align*}
gives the result where $X_{ij} \leq d$ for all $(i,j).$
\end{proof}

\begin{lemma}[Schur lemma and trace inequalities]\label{lem:alg}
For two positive semidefinite matrices $\BX$ and $\BY$,
it holds
\begin{equation}\label{eq:alg1}
\BX\circ\BY\succeq 0,~~~\lag \BX,\BY\rag \leq \|\BX\|~\|\BY\|_*
\end{equation}
where $``\circ"$ denotes the Hadamard product.

In addition, assume $\BX$ has a null space of dimension $d$ that belongs to the null space of $\BY$. In particular, let $\BPi$ be the corresponding orthogonal projection matrix of the null space of $\BX$, and it satisfies $\BY\BPi = 0$. Then it holds
\begin{equation}\label{eq:alg2}
\lag \BX,\BY\rag \geq \lambda_{d+1}(\BX) \Tr((\I_n - \BPi)\BY)
\end{equation}
where $\lambda_{d+1}(\BX)$ is the $(d+1)$-th smallest eigenvalue of $\BX.$
\end{lemma}
\begin{proof}
The first equation~\eqref{eq:alg1} follows from directly Schur lemma, and Von Neumann's trace inequality. It is also straightforward to show~\eqref{eq:alg2}: let $\bu_i$ be the eigenvector of $(\I_n - \BPi)\BY(\I_n - \BPi)$ with eigenvalue $\lambda_i$, i.e., $(\I_n - \BPi)\BY(\I_n - \BPi) = \sum_i \lambda_i \bu_i\bu_i^{\top}$ with $\lambda_i \geq 0$. Then
\begin{align*}
\lag \BX,\BY\rag & = \sum_i \lambda_i \lag \BX, \bu_i\bu_i^{\top}\rag  = \sum_i \lambda_i\lag \BX, (\I_n -\BPi) \bu_i\bu_i^{\top}(\I_n -\BPi)\rag \\
& \geq \lambda_{d+1}(\BX) \sum_i \lambda_i \|(\I_n - \BPi)\bu_i\|^2 \geq \lambda_{d+1}(\BX) \sum_i \lambda_i \|\bu_i\|^2 \\
& = \lambda_{d+1}(\BX)\Tr( (\I_n - \BPi)\BY(\I_n - \BPi) )
\end{align*}
where $\BPi\bu_i= 0$ and $\BX\BPi = 0.$
\end{proof}

\begin{lemma}\label{lem:gauss}
For any $\BPhi\in\RR^{d\times p}$ Gaussian random matrix, i.e., all entries are i.i.d. standard normal, it holds:
\begin{align*}
\E(\BPhi\BU^{\top}\BV\BPhi^{\top}) = \lag \BU, \BV\rag \I_d, \qquad \E(\BPhi\BU^{\top}\BPhi\BV^{\top}) = \BU\BV^{\top}
\end{align*}
for any $\BU$ and $\BV\in\RR^{d\times p}$ that are independent of $\BPhi.$
\end{lemma}
Lemma~\ref{lem:gauss} is straightforward to verify by using independence of all entries in $\BPhi.$ With the three lemmas above, we are ready to present the proof of Theorem~\ref{thm:main2}. In fact, the chosen random direction is similar to the one presented in~\cite{MB23}, and the key here is to provide a tight bound on the resulting quadratic form that only depends on the condition number of the ``Laplacian" $\BL.$

\begin{proof}[\bf Proof of Theorem~\ref{thm:main2}]
Assume $\widehat{\BO}_i\in\Od(d)$ is the global minimizer to~\eqref{def:ls} and $\widehat{\BO}\widehat{\BO}^{\top}$ satisfies the optimality condition of~\eqref{def:sdr}, i.e.,
\[
\widehat{\BLambda} = \BDG(\BA\widehat{\BO}\widehat{\BO}^{\top}),~~(\widehat{\BLambda} - \BA) \widehat{\BO} = 0,~~~\widehat{\BLambda} - \BA\succeq 0,
\]
where $\widehat{\BLambda}_{ii}  = \sum_{j=1}^n\BA_{ij}\widehat{\BO}_j\widehat{\BO}_i^{\top}$ is symmetric.
From now on, we denote $\BL := \widehat{\BLambda} - \BA\succeq 0$.
For any SOCP $\BS\in\St(p,d)^{\otimes n}$, it satisfies:
\[
\lag \BDG(\BA\BS\BS^{\top}) - \BA, \dot{\BS}\dot{\BS}^{\top} \rag \geq 0,~~~~\sum_{j=1}^n \BA_{ij}\BS_j\BS_i^{\top} \text{ is symmetric}
\] 
where $\dot{\BS}$ is in the tangent space of Stiefel manifold at $\BS.$
Note that
\begin{align*}
\BDG(\BA\BS\BS^{\top}) - \BA & = \BDG((\widehat{\BLambda} - \BL)\BS\BS^{\top}) - (\widehat{\BLambda} - \BL) \\
& = \BL-\BDG(\BL\BS\BS^{\top}) + \BDG(\widehat{\BLambda} \BS\BS^{\top}) - \widehat{\BLambda} \\
& = \BL - \BDG(\BL\BS\BS^{\top})
\end{align*}
where $[\BDG(\widehat{\BLambda}\BS\BS^{\top})]_{ii} = \widehat{\BLambda}_{ii}\BS_i\BS_i^{\top} = \widehat{\BLambda}_{ii}$ implies $\BDG(\widehat{\BLambda}\BS\BS^{\top}) = \widehat{\BLambda}$. Therefore, it holds that
\begin{equation}\label{eq:socp}
\lag \BL-\BDG(\BL\BS\BS^{\top}), \dot{\BS}\dot{\BS}^{\top} \rag \geq 0
\end{equation}
for any second order critical point $\BS.$

We choose a random direction:
\[
\dot{\BS}_i = \widehat{\BO}_i\BPhi - \BS_i \BPhi^{\top} \widehat{\BO}_i^{\top}\BS_i\in\RR^{d\times p}
\]
where $\BPhi\in\RR^{d\times p}$ is a Gaussian random independent of $\widehat{\BO}$. {Similar approaches for selecting random directions have been proposed in~\cite{MTG20,MB23,MMMO17,L23}. The underlying intuition is to choose a random direction at $\BS$ that points toward the global minimizer $\widehat{\BO}$.}
It is easy to verify that $\dot{\BS}_i$ is in the tangent space of $\BS_i\in\St(p,d)$:
\begin{align*}
\dot{\BS}_i\BS_i^{\top} + \BS_i \dot{\BS}_i^{\top} & = ( \widehat{\BO}_i\BPhi - \BS_i \BPhi^{\top} \widehat{\BO}_i^{\top}\BS_i)\BS_i^{\top} + \BS_i (\BPhi^{\top}\widehat{\BO}_i^{\top} - \BS_i^{\top}\widehat{\BO}_i \BPhi\BS_i^{\top}) \\
& = \widehat{\BO}_i\BPhi\BS_i^{\top} - \BS_i\BPhi^{\top}\widehat{\BO}_i^{\top} +\BS_i\BPhi^{\top}\widehat{\BO}_i^{\top} - \widehat{\BO}_i\BPhi\BS_i^{\top} = 0.
\end{align*}
Taking expectation w.r.t. the Gaussian random vector $\BPhi$ leads to
\begin{align*}
\E\dot{\BS}_i \dot{\BS}_j^{\top} & = \E ( \widehat{\BO}_i\BPhi - \BS_i \BPhi^{\top} \widehat{\BO}_i^{\top}\BS_i)( \widehat{\BO}_j\BPhi - \BS_j \BPhi^{\top} \widehat{\BO}_j^{\top}\BS_j)^{\top} \\
& = \E \left( \widehat{\BO}_i\BPhi\BPhi^{\top}\widehat{\BO}_j^{\top} 
- \BS_i\BPhi^{\top}\widehat{\BO}_i^{\top}\BS_i \BPhi^{\top}\widehat{\BO}_j^{\top}  
- \widehat{\BO}_i\BPhi \BS_j^{\top}\widehat{\BO}_j\BPhi\BS_j^{\top} 
+ \BS_i\BPhi^{\top}\widehat{\BO}_i^{\top}\BS_i \BS_j^{\top}\widehat{\BO}_j\BPhi\BS_j^{\top}\right) \\
& = p\widehat{\BO}_i\widehat{\BO}_j^{\top} - \BS_i \BS_i^{\top}\widehat{\BO}_i\widehat{\BO}_j^{\top} -\widehat{\BO}_i\widehat{\BO}_j^{\top}  \BS_j\BS_j^{\top} + \lag \widehat{\BO}_i^{\top}\BS_i,\widehat{\BO}_j^{\top}\BS_j\rag\BS_i\BS_j^{\top} \\
& = (p-2)\widehat{\BO}_i\widehat{\BO}_j^{\top} + \lag \widehat{\BO}_i^{\top}\BS_i, \widehat{\BO}_j^{\top}\BS_j\rag\BS_i\BS_j^{\top}
\end{align*}
following from Lemma~\ref{lem:gauss} and $\BS_i\BS_i^{\top} = \I_d.$

Now we compute $\lag \BL - \BDG(\BL\BS\BS^{\top}), \E\dot{\BS}\dot{\BS}^{\top}\rag$ in~\eqref{eq:socp}:
\begin{equation}\label{eq:socp1pre}
\begin{aligned}
\lag \BL,\E\dot{\BS} \dot{\BS}^{\top}\rag & = (p-2)\lag \BL,\widehat{\BO}\widehat{\BO}^{\top}\rag  + \sum_{i,j} \lag \widehat{\BO}_i^{\top}\BS_i, \widehat{\BO}_j^{\top}\BS_j\rag \lag \BL_{ij}, \BS_i\BS_j^{\top}\rag \\
& = \sum_{i,j} \lag \widehat{\BO}_i^{\top}\BS_i, \widehat{\BO}_j^{\top}\BS_j\rag \lag \BL_{ij}, \BS_i\BS_j^{\top}\rag
\end{aligned}
\end{equation}
where $\BL \widehat{\BO} = 0$. Define 
\begin{equation}\label{def:X}
X_{ij} = \lag \widehat{\BO}_i^{\top}\BS_i, \widehat{\BO}_j^{\top}\BS_j\rag = \lag \BS_i\BS_j^{\top}, \widehat{\BO}_i\widehat{\BO}_j^{\top} \rag,
\end{equation} and apparently $\BX\in\RR^{n\times n}$ is positive semidefinite  with $X_{ii} = d.$ Then it holds
\begin{align*}
\lag \BL,\E\dot{\BS}\dot{\BS}^{\top}\rag & = \sum_{i,j}X_{ij} \lag \BL_{ij}, \BS_i\BS_j^{\top}\rag = \lag \BL, (\BX\otimes \BJ_d)\circ \BS\BS^{\top} \rag \\
(\text{Use }\BL\widehat{\BO} = 0)~~~~  & = \lag \BL, (\I_{nd} - n^{-1}\widehat{\BO}\widehat{\BO}^{\top}) ((\BX\otimes \BJ_d)\circ \BS\BS^{\top})(\I_{nd} - n^{-1}\widehat{\BO}\widehat{\BO}^{\top}) \rag \\
(\text{Lemma~\ref{lem:alg}})~~~~ & \leq \lambda_{\max}(\BL)\cdot \|(\I_{nd} - n^{-1}\widehat{\BO}\widehat{\BO}^{\top}) ((\BX\otimes \BJ_d)\circ \BS\BS^{\top})(\I_{nd} - n^{-1}\widehat{\BO}\widehat{\BO}^{\top})\|_*
\end{align*}

Note that $(\BX\otimes \BJ_d)\circ \BS\BS^{\top}$ is positive semidefinite, and then the nuclear norm equals
\begin{align*}
& \Tr\left[(\I_{nd} - n^{-1}\widehat{\BO}\widehat{\BO}^{\top}) ((\BX\otimes \BJ_d)\circ \BS\BS^{\top})(\I_{nd} - n^{-1}\widehat{\BO}\widehat{\BO}^{\top})\right] \\
& = \lag (\BX\otimes \BJ_d)\circ \BS\BS^{\top}, \I_{nd} - n^{-1}\widehat{\BO}\widehat{\BO}^{\top}\rag \\
& = \Tr( (\BX\otimes \BJ_d)\circ \BS\BS^{\top}) - n^{-1} \lag (\BX\otimes \BJ_d)\circ \BS\BS^{\top} , \widehat{\BO}\widehat{\BO}^{\top}\rag \\
& = \sum_{i=1}^n X_{ii}\Tr(\BS_i\BS_i^{\top}) - n^{-1}\sum_{i,j}X_{ij} \lag \BS_i\BS_j^{\top}, \widehat{\BO}_i\widehat{\BO}_j^{\top} \rag = nd^2 - n^{-1}\|\BX\|_F^2
\end{align*}
where $X_{ii} = d$ and $(\BX\otimes \BJ_d)\circ\BS\BS^{\top}$ is positive semidefinite whose diagonal blocks are $d\I_d.$ Therefore, we have
\begin{equation}\label{eq:socp1}
\lag \BL, \E\dot{\BS}\dot{\BS}^{\top}\rag \leq \lambda_{\max}(\BL)\cdot(nd^2 - n^{-1}\|\BX\|_F^2).
\end{equation}

For the second term $\BDG(\BL\BS\BS^{\top})$ in~\eqref{eq:socp}, it holds 
\begin{align*}
\lag \BDG(\BL\BS\BS^{\top}),\E\dot{\BS}\dot{\BS}^{\top}\rag & = (p+d-2)\lag\BDG(\BL\BS\BS^{\top}), \I_{nd} \rag \\
& = (p+d-2) \sum_{i=1}^n\sum_{j=1}^n \lag \BL_{ij}, \BS_i\BS_j^{\top} \rag = (p+d-2)\lag \BL,\BS\BS^{\top}\rag
\end{align*}
where $[\E \dot{\BS}_i\dot{\BS}_i^{\top}]_{ii} = (p+d-2)\I_d.$
Note that
\begin{equation}\label{eq:socp2}
\begin{aligned}
\lag \BL,\BS\BS^{\top}\rag & = \lag \BL,(\I_{nd} - n^{-1}\widehat{\BO}\widehat{\BO}^{\top})\BS\BS^{\top}(\I_{nd} - n^{-1}\widehat{\BO}\widehat{\BO}^{\top})\rag \\
(\text{Lemma~\ref{lem:alg}})~~~& \geq  \lambda_{d+1}(\BL)\cdot \lag\BS\BS^{\top}, \I_{nd} - n^{-1}\widehat{\BO}\widehat{\BO}^{\top} \rag \\
& = \lambda_{d+1}(\BL)\cdot \left(nd - \frac{\|\widehat{\BO}^{\top}\BS\|_F^2}{n}\right) \\
& = \lambda_{d+1}(\BL)\cdot \left(nd - \frac{\lag \BX,\BJ_n\rag}{n}\right)
\end{aligned}
\end{equation}
where $\|\widehat{\BO}^{\top}\BS\|_F^2 = \|\sum_{i=1}^n\widehat{\BO}_i^{\top}\BS_i \|_F^2 = \sum_{i,j}\lag\widehat{\BO}_i^{\top}\BS_i,\widehat{\BO}_j^{\top}\BS_j \rag = \lag \BX,\BJ_n\rag.$

As a result, combining~\eqref{eq:socp1} with~\eqref{eq:socp2} gives 
\begin{align}
& \lag \BL - \BDG(\BL\BS\BS^{\top}), \E\dot{\BS}\dot{\BS}^{\top}\rag \nonumber \\
& \leq  \lambda_{\max}(\BL)\cdot (nd^2 - n^{-1} \|\BX\|_F^2 )  - (p+d-2) \lambda_{d+1}(\BL)\cdot (nd - n^{-1}\lag \BX,\BJ_n\rag) \nonumber \\
& \leq (2d\lambda_{\max}(\BL) - (p+d-2)\lambda_{d+1}(\BL))\cdot (nd - n^{-1}\lag \BX,\BJ_n\rag) \label{eq:mainthm_final}
\end{align}
following from Lemma~\ref{lem:Xineq}. Note that $\lag \BX,\BJ_n\rag\leq nd^2$. Suppose 
\[
p+d-2\geq \frac{2d\lambda_{\max}(\BL)}{\lambda_{d+1}(\BL)},
\]
then $\lag \BL - \BDG(\BL\BS\BS^{\top}), \E\dot{\BS}\dot{\BS}^{\top}\rag\leq 0$. As every second order critical point satisfies~\eqref{eq:socp}, this implies that at any local minimizer, it must hold $\lag \BX,\BJ_n\rag = n^2 d$ which gives $\BX = d\BJ_n$ since $\BX$ is positive semidefinite and the diagonal entries are $d$. Thus $X_{ij} = \lag \widehat{\BO}_i^{\top}\BS_i, \widehat{\BO}_j^{\top}\BS_j\rag =d$ indicates that $\widehat{\BO}_i^{\top}\BS_i = \widehat{\BO}_j^{\top}\BS_j$ and $\widehat{\BO}\widehat{\BO}^{\top} = \BS\BS^{\top}$. Note that the proof does not assume any assumption on $d$, which also holds for $d=1$ automatically.
\end{proof}

\begin{proof}[\bf Proof of Theorem~\ref{thm:adv}: robustness against monotone adversaries] 
Consider $\BA = \BA_M + \BDelta$ where $\BA_M$ is the original data matrix and $\BDelta$ represents the additive monotone adversaries. 
To study the landscape, it suffices to consider the Laplacian associated with $\BA_M$ and $\BDelta$ respectively:
\[
\BL = \underbrace{\BDG(\BA_M\bx\bx^{\top}) - \BA_M}_{\BL_M} + \underbrace{\BDG( \BDelta \bx\bx^{\top}) - \BDelta}_{\BL_{\Delta}}
\]
where $\BL_M\bx = \BL_{\Delta}\bx = 0$ and $\BDelta\circ \bx\bx^{\top}\geq 0$ holds since $\BDelta$ is the ``favorable" adversary.

We pick up the random direction in the tangent space $\dot{\BS}_i = x_i (\BPhi - \BS_i \BPhi^{\top}\BS_i)$ where $\BPhi\in\RR^{1\times p}$ is a standard Gaussian random vector and $\BS_i\in\RR^{1\times p}$ satisfies $\BS_i\BS_i^{\top}=1$. 
The idea is to analyze the quadratic form associated with $\BL_M$ and $\BL_{\Delta}$, and show that they have a common negative curvature direction. 

By letting $d=1$ and $\BX= \BS\BS^{\top}\circ \bx\bx^{\top}$ in~\eqref{def:X}, the equation~\eqref{eq:mainthm_final} implies
\begin{equation}\label{eq:adv1}
\lag \BL_M - \BDG(\BL_M\BS\BS^{\top}), \E\dot{\BS}\dot{\BS}^{\top}\rag  \leq (2\lambda_{\max}(\BL_M) - (p-1)\lambda_{2}(\BL_M))\cdot (n - n^{-1}\lag \BX,\BJ_n\rag) \leq 0.
\end{equation}

For the second part $\lag \BL_{\Delta}-\BDG(\BL_{\Delta}\BS\BS^{\top}),\E\dot{\BS}\dot{\BS}^{\top} \rag$, we adopt the technique in~\cite{MB23}. The key part is to estimate $\lag \BL_{\Delta}, \E \dot{\BS}\dot{\BS}^{\top}\rag$: the equation~\eqref{eq:socp1pre} with $\widehat{\BO}_i = x_i$ implies
\begin{align*}
\lag \BL_{\Delta}, \E \dot{\BS}\dot{\BS}^{\top}\rag & = \sum_{i,j} L_{\Delta,ij} x_ix_j|\lag \BS_i, \BS_j\rag|^2.
\end{align*}
Note that
\[
x_ix_j\lag \BS_i, \BS_j\rag = 1 - \frac{1}{2}\|x_i\BS_i - x_j\BS_j\|^2 
\]
where $x_i\in\{\pm 1\}$ and thus
\[
|\lag \BS_i, \BS_j\rag|^2 = 1 - \|x_i\BS_i - x_j\BS_j\|^2 + \frac{1}{4}\|x_i\BS_i - x_j\BS_j\|^4 = - 1 + 2x_ix_j\lag \BS_i,\BS_j\rag +  \frac{1}{4}\|x_i\BS_i - x_j\BS_j\|^4.
\]
This gives rise to 
\begin{align*}
\lag \BL_{\Delta}, \E \dot{\BS}\dot{\BS}^{\top}\rag & =2 \sum_{i,j} x_ix_jL_{\Delta,ij} \cdot x_ix_j\lag\BS_i,\BS_j \rag + \frac{1}{4}\sum_{i\neq j} x_i x_j L_{\Delta,ij}  \|x_i\BS_i -x_j \BS_j\|^4 \\
& = 2\lag \BL_{\Delta}, \BS\BS^{\top}\rag -\frac{1}{4} \sum_{i,j}x_ix_j\Delta_{ij} \|x_i\BS_i -x_j \BS_j\|^4
\end{align*}
where $L_{\Delta,ij} = -\Delta_{ij}$ for $i\neq j$ and $\BL_{\Delta}\bx = 0.$
On the other hand, we have
\[
\lag \BDG(\BL_{\Delta}\BS\BS^{\top}), \dot{\BS}\dot{\BS}^{\top}\rag = (p-1) \lag \BL_{\Delta}, \BS\BS^{\top}\rag
\]
where $\E \dot{\BS}_i\dot{\BS}_i^{\top} = p-1.$

Therefore, we have
\begin{equation}\label{eq:adv2}
\begin{aligned}
& \lag \BL_{\Delta} -\BDG(\BL_{\Delta}\BS\BS^{\top}), \E \dot{\BS}\dot{\BS}^{\top}\rag  \\
& \qquad = 2 \lag \BL_{\Delta}, \BS\BS^{\top}\rag - \frac{1}{4}\sum_{i,j} x_ix_j\Delta_{ij} \| x_i\BS_i - x_j\BS_j \|_F^4 - (p-1)\lag \BL_{\Delta},\BS\BS^{\top}\rag \\
& \qquad = (3 - p) \lag \BL_{\Delta},\BS\BS^{\top}\rag -  \frac{1}{4}\sum_{i,j} x_ix_j \Delta_{ij}\| x_i\BS_i - x_j\BS_j \|_F^4 \leq 0
\end{aligned}
\end{equation}
for any $p\geq 3$ since $x_ix_j \Delta_{ij} \geq 0$ holds for any $i\neq j$ and $\BL_{\Delta}\succeq 0$.

Therefore, provided that
\[
2\lambda_{\max}(\BL_M) - (p-1)\lambda_2(\BL_M) \leq 0,~~p\geq 3,
\]
i.e., $p \geq 2\lambda_{\max}(\BL_M)/ \lambda_2(\BL_M)+1,$ 
we have
$\lag \BL-  \BDG(\BL \BS\BS^{\top}), \E\dot{\BS}\dot{\BS}^{\top}\rag \leq 0$
which follows from~\eqref{eq:adv1} and~\eqref{eq:adv2}.
However, for any second order critical point, the quadratic form above is nonnegative. Therefore, it is only possible if the value equals 0 which gives $x_i \BS_i = x_j \BS_j$, i.e., $\bx\bx^{\top} = \BS\BS^{\top}.$ In other words, the only local minimizer is given by $\BS$ satisfying $\BS\BS^{\top} = \bx\bx^{\top}$ which is also the global minimizer.
\end{proof}

\subsection{Proof of Corollary~\ref{cor:odsync} and~\ref{cor:gopp}}\label{ss:cor}

\begin{proof}[\bf Proof of Corollary~\ref{cor:odsync}: orthogonal group synchronization with Gaussian noise]
Theorem 3.2 and 4.1 in~\cite{L22} state that if $\sigma < c_0\sqrt{n}/(\sqrt{d}(\sqrt{d} + \sqrt{\log n}))$ for some constant $c_0>0$, then with high probability, the global minimizer $\widehat{\BO}$ to the SDR~\eqref{def:sdr} satisfies
\begin{align}
& \min_{\BQ\in\Od(d)}\|\widehat{\BO} - \BO\BQ\|_F \leq  \eps \sqrt{ nd}, \label{eq:o_error} \\
&\max_{1\leq i\leq n}\|\BW_i^{\top}\widehat{\BO}\|_F \leq \xi\sqrt{nd}(\sqrt{d} + 4\sqrt{\log n}), \nonumber
\end{align}
where  $\xi = 2+3\eps$ and 
\[
\eps \sqrt{nd} \lesssim \sigma\sqrt{d} (\sqrt{d} + 4\sqrt{\log n}) \lesssim \sqrt{n}.
\]
Here the minimizer $\BQ$ to~\eqref{eq:o_error} is given by the matrix sign of $\BO^{\top}\widehat{\BO}$, i.e., $\BQ$ equals the product of left and right singular vectors of $\BO^{\top}\widehat{\BO}$ and
\[
\|\widehat{\BO} - \BO\BQ\|_F^2 = 2nd - 2\lag \widehat{\BO}, \BO\BQ\rag = 2nd - 2\|\BO^{\top}\widehat{\BO}\|_*.
\]

Denote $\BLambda = \BDG(\BA\widehat{\BO}\widehat{\BO}^{\top})$ and we have
\begin{align*}
\BL & = \BLambda - \BA=  (\I_{nd} - n^{-1}\widehat{\BO}\widehat{\BO}^{\top} )(\BLambda - \BA) (\I_{nd} - n^{-1}\widehat{\BO}\widehat{\BO}^{\top} ).
\end{align*}
The top and $(d+1)$-th smallest eigenvalue of $\BL$ are bounded by
\begin{align*}
\lambda_{\max}(\BL) & \leq \lambda_{\max}(\BLambda) + \|  (\I_{nd} - n^{-1}\widehat{\BO}\widehat{\BO}^{\top} )  \BA  (\I_{nd} - n^{-1}\widehat{\BO}\widehat{\BO}^{\top} ) \|, \\
\lambda_{d+1}(\BL) & \geq \lambda_{\min}(\BLambda) - \|  (\I_{nd} - n^{-1}\widehat{\BO}\widehat{\BO}^{\top} )  \BA  (\I_{nd} - n^{-1}\widehat{\BO}\widehat{\BO}^{\top} ) \|, 
\end{align*}
and
\begin{align*}
\|  (\I_{nd} - n^{-1}\widehat{\BO}\widehat{\BO}^{\top} )  \BA  (\I_{nd} - n^{-1}\widehat{\BO}\widehat{\BO}^{\top} ) \| \leq \| (\I_{nd} - n^{-1}\widehat{\BO}\widehat{\BO}^{\top})\BO\|_F^2 + \sigma \|\BW\|
\end{align*}
where $\BA = \BO\BO^{\top} + \sigma\BW.$

Note that at the global minimizer, $\BLambda_{ii} = \sum_{j=1}^n \BA_{ij}\widehat{\BO}_j\widehat{\BO}_i^{\top}$ is symmetric and  positive semidefinite. Therefore, its eigenvalues match the singular values. 
For the $i$-th block, it holds 
\[
\lambda_{\min}(\BLambda_{ii}) = \sigma_{\min}\left( \BO_i\BO^{\top}\widehat{\BO}\widehat{\BO}^{\top}_i  + \sigma \BW_i^{\top}\widehat{\BO}\widehat{\BO}_i^{\top} \right)
\]
where $\BW_i$ is the $i$-th block column of $\BW.$
The smallest and largest eigenvalues of $\BLambda$ satisfy
\begin{align*}
\lambda_{\max}(\BLambda) & \leq \sigma_{\max}(\BO^{\top}\widehat{\BO}) + \sigma \max_{1\leq i\leq n}\|\BW_i^{\top}\widehat{\BO}\|, \\
\lambda_{\min}(\BLambda) & \geq \sigma_{\min}(\BO^{\top}\widehat{\BO}) - \sigma \max_{1\leq i\leq n}\|\BW_i^{\top}\widehat{\BO}\|.
\end{align*}
Then we have
\begin{align*}
\lambda_{\max}(\BL) & \leq \sigma_{\max}(\BO^{\top}\widehat{\BO}) + \sigma \max_{1\leq i\leq n}\|\BW_i^{\top}\widehat{\BO}\| + \| (\I_{nd} - n^{-1}\widehat{\BO}\widehat{\BO}^{\top})\BO\|_F^2 + \sigma \|\BW\|,\\
\lambda_{d+1}(\BL) & \geq \sigma_{\min}(\BO^{\top}\widehat{\BO}) - \sigma \max_{1\leq i\leq n}\|\BW_i^{\top}\widehat{\BO}\| - \| (\I_{nd} - n^{-1}\widehat{\BO}\widehat{\BO}^{\top})\BO\|_F^2 - \sigma \|\BW\|.
\end{align*}

Note that
\[
\|\widehat{\BO} - \BO\BQ\|_F \leq  \eps \sqrt{nd} \Longleftrightarrow 
2 nd - 2\|\BO^{\top}\widehat{\BO}\|_* \leq \eps^2 nd \Longrightarrow \sigma_{\min}(\BO^{\top}\widehat{\BO}) \geq \left(1-\frac{\eps^2d}{2}\right)n
\]
where all the $d$ singular values of $\BO^{\top}\widehat{\BO}$ are no larger than $n$.
This gives
\[
 \left(1-\frac{\eps^2d}{2}\right)n \leq \sigma_{\min}(\BO^{\top}\widehat{\BO}) \leq \sigma_{\max}(\BO^{\top}\widehat{\BO}) \leq n.
\]
Moreover, it holds
\begin{equation}\label{eq:ohato}
 \| (\I_{nd} - n^{-1}\widehat{\BO}\widehat{\BO}^{\top})\BO\|_F =  \| (\I_{nd} - n^{-1}\widehat{\BO}\widehat{\BO}^{\top})(\BO - \widehat{\BO} \BQ^{\top})\|_F \leq \|\widehat{\BO} - \BO\BQ\|_F.
\end{equation}

Then we have
\begin{align*}
\lambda_{d+1}(\BL) & \geq \left( 1- \frac{\eps^2 d}{2}\right)n - \xi\sigma\sqrt{nd}(\sqrt{d}+4\sqrt{\log n}) - (\|\widehat{\BO} - \BO\BQ\|_F^2 + 3\sigma\sqrt{nd}) \\
& \geq \left( 1- \frac{3\eps^2d}{2}\right)n - (3+\xi) \sigma \sqrt{nd}(\sqrt{d} + 4\sqrt{\log n}) \\
& \geq n - C \sigma \sqrt{nd}(\sqrt{d} + 4\sqrt{\log n})
\end{align*}
for some $C>0$
where $\|\BW\| \leq 3\sqrt{nd}$ and
\[
\eps^2 nd \lesssim \sqrt{n}\cdot \eps\sqrt{nd} \lesssim \sigma\sqrt{nd}(\sqrt{d} + 4\sqrt{\log n}).
\]
It means
\[
n - \lambda_{d+1}(\BL)\leq C \sigma\sqrt{nd}(\sqrt{d} + 4\sqrt{\log n})
\]
for some constant $C>0.$ Similarly, using the estimation above gives $\lambda_{\max}(\BL) \leq n+C \sigma\sqrt{nd}(\sqrt{d} + 4\sqrt{\log n})$. Finally, applying Theorem~\ref{thm:main2} leads to a sufficient condition for the benign landscape:
\[
\frac{\lambda_{\max}(\BL)}{\lambda_{d+1}(\BL)} \leq \frac{n+C \sigma\sqrt{nd}(\sqrt{d} + 4\sqrt{\log n})}{n-C \sigma\sqrt{nd}(\sqrt{d} + 4\sqrt{\log n})} \leq \frac{p+d-2}{2d}
\]
and thus Corollary~\ref{cor:odsync} follows.
\end{proof}

\begin{proof}[\bf Proof of Corollary~\ref{cor:gopp}: generalized orthogonal Procrustes problem]

We write $\BA$ whose $(i,j)$-block equals $\BA_i\BA_j^{\top}$ into the matrix form:
\[
\BA = \BO\bar{\BA}\bar{\BA}^{\top}\BO^{\top} + \sigma (\BO\bar{\BA}\BW^{\top} + \BW \bar{\BA}^{\top}\BO^{\top} + \sigma \BW\BW^{\top})
\]
where $\BW\in\RR^{nd\times m}$ is a Gaussian random matrix satisfying 
\[
\|\BW\| \leq \sqrt{nd} + \sqrt{m} + \sqrt{2\gamma n\log n}
\]
with probability at least $1 - O(n^{-\gamma+1})$, following from~\cite[Theorem 2.26]{W19} and~\cite[Theorem 7.3.1]{V18}.
This data matrix $\BA$ is decomposed into the signal plus noise form $\BA =  \BO\bar{\BA}\bar{\BA}^{\top}\BO^{\top}  + \sigma\BDelta$:
$\BDelta = \BO\bar{\BA}\BW^{\top} + \BW \bar{\BA}^{\top}\BO^{\top} + \sigma \BW\BW^{\top}$ and
\begin{equation}\label{eq:delta_gopp}
\|\BDelta\|\leq 3\sqrt{n}\|\bar{\BA}\|\|\BW\|
\end{equation}
under the assumption on $\sigma:$
\begin{equation}\label{eq:sigma_gopp}
\sigma \lesssim \frac{\sigma_{\min}(\bar{\BA})}{\kappa^4}\cdot\frac{\sqrt{n}}{\sqrt{d}(\sqrt{nd}+\sqrt{m}+\sqrt{2\gamma n\log n})}.
\end{equation}
Moreover, under~\eqref{eq:sigma_gopp}, Theorem 3.2 in~\cite{L23b} guarantees the SDR~\eqref{def:sdr} is tight, i.e., the global minimizer equals $\widehat{\BO}\widehat{\BO}^{\top}$ for some $\widehat{\BO}\in\Od(d)^{\otimes n}$.
In addition, Theorem 3.3 in~\cite{L23b} implies
\[
\min_{\BQ\in \Od(d)}\|\widehat{\BO} - \BO\BQ\|_F \leq \eps\sqrt{nd}
\]
and
\begin{equation}\label{eq:eps_gopp}
\eps\sqrt{nd} \lesssim \frac{\kappa^2}{\|\bar{\BA}\|}\cdot \sigma \sqrt{d}(\sqrt{nd} + \sqrt{m} + \sqrt{2\gamma n\log n}) 
\lesssim \frac{\sqrt{n}}{\kappa^3}
\end{equation}
under the assumption on $\sigma$ in~\eqref{eq:sigma_gopp}.

Denote $\BLambda = \BDG(\BA\widehat{\BO}\widehat{\BO}^{\top})$ and $\BL = \BLambda - \BA$, and $\BL$ satisfies $\BL\widehat{\BO} = 0.$ At the global minimizer, $\BLambda_{ii} = \sum_{j=1}^n \BA_j\widehat{\BO}_j\widehat{\BO}_i^{\top}$ is symmetric, and its smallest eigenvalue is lower bounded by
\[
n\|\bar{\BA}\|^2 \geq \lambda_{\min}(\BLambda) \geq \frac{1-\eps^2 d}{\kappa^2}\cdot n\|\bar{\BA}\|^2
\]
which follows from the proof of~\cite[Theorem 3.2, pp. 29 and Proposition 5.6]{L23b}.

We have
\begin{align*}
\lambda_{\max}(\BL) & \leq \lambda_{\max}(\BLambda) + \|  \left(\I_{nd} - n^{-1}\widehat{\BO}\widehat{\BO}^{\top}\right) \BA\left(\I_{nd} - n^{-1}\widehat{\BO}\widehat{\BO}^{\top}\right) \|, \\
\lambda_{d+1}(\BL) & \geq \lambda_{\min}(\BLambda) - \|  \left(\I_{nd} - n^{-1}\widehat{\BO}\widehat{\BO}^{\top}\right) \BA\left(\I_{nd} - n^{-1}\widehat{\BO}\widehat{\BO}^{\top}\right) \|, 
\end{align*}
where $\BA = \BO\bar{\BA}\bar{\BA}^{\top}\BO^{\top} +\sigma \BDelta$. Here the proof of Theorem 3.2 in~\cite{L23b} gives
\begin{align*}
& \|  \left(\I_{nd} - n^{-1}\widehat{\BO}\widehat{\BO}^{\top}\right) \BA\left(\I_{nd} - n^{-1}\widehat{\BO}\widehat{\BO}^{\top}\right) \| 
 \leq \|\left(\I_{nd} - n^{-1}\widehat{\BO}\widehat{\BO}^{\top}\right)  \BO\bar{\BA}\|_F^2 + \sigma\|\BDelta\| \\
& = \|\bar{\BA}\|^2 \|\widehat{\BO} - \BO\BQ\|_F^2 + 3\sigma\sqrt{n}\|\bar{\BA}\|\|\BW\| \\
& \leq \eps^2 nd \|\bar{\BA}\|^2 + 3\sigma \sqrt{n} \|\bar{\BA}\| (\sqrt{nd} + \sqrt{m} + \sqrt{2\gamma n\log n}) %\\
%& \leq n\|\bar{\BA}\|^2 \cdot \frac{4\kappa^2 \sigma \sqrt{d} (\sqrt{nd} + \sqrt{m} + \sqrt{2\gamma n\log n})}{\sqrt{n}\|\bar{\BA}\|}
\end{align*}
which follows from~\eqref{eq:ohato},~\eqref{eq:delta_gopp}, and~\eqref{eq:eps_gopp}.

The $(d+1)$-th smallest eigenvalue is lower bounded by
\begin{align*}
\lambda_{d+1}(\BL) & \geq n\|\bar{\BA}\|^2 \left( \frac{1-\eps^2 d}{\kappa^2} - \left(\eps^2d + \frac{3\sigma (\sqrt{nd} + \sqrt{m} + \sqrt{2\gamma n\log n})}{\sqrt{n}\|\bar{\BA}\|}\right)\right) \\
& \geq n\|\bar{\BA}\|^2 \left( \frac{1}{\kappa^2} - \frac{5C_0\kappa^2 \sigma \sqrt{d} (\sqrt{nd} + \sqrt{m} + \sqrt{2\gamma n\log n})}{\sqrt{n}\|\bar{\BA}\|} \right)
\end{align*}
for some constant $C_0 > 0$
where 
\[
\eps^2 nd \lesssim \frac{\kappa^2}{\|\bar{\BA}\|}\cdot \sigma \sqrt{d}(\sqrt{nd} + \sqrt{m} + \sqrt{2\gamma n\log n})  \cdot \sqrt{n} \Longleftrightarrow 
\eps^2 d \lesssim \frac{\kappa^2\sigma \sqrt{d}(\sqrt{nd} + \sqrt{m} + \sqrt{2\gamma n\log n})  }{\sqrt{n}\|\bar{\BA}\|}.
\]

%Here
%\begin{align*}
%\left\| \sum_{j=1}^n \BC_{ij}\widehat{\BO}_j \right\| & \leq \|\BO_i \bar{\BA}\bar{\BA}^{\top}\BO^{\top}\widehat{\BO}\| + \sigma \|\BDelta_i^{\top}\widehat{\BO} \| \leq n\|\bar{\BA}\|^2 + \sigma \|\BDelta_i^{\top}\widehat{\BO}\|
%\end{align*}
%We know that
%\[
%\|\BDelta_i^{\top}\widehat{\BO}\| \leq 3\sqrt{nd}(\sqrt{nd}+\sqrt{m}+\sqrt{2\gamma n\log n})\|\bar{\BA}\|.
%\]
For the largest eigenvalue of $\BL$, it holds that
\begin{align*}
\lambda_{\max}(\BL) & \leq n\|\bar{\BA}\|^2 + \eps^2 nd \|\bar{\BA}\|^2 + 3\sigma \sqrt{n} \|\bar{\BA}\| (\sqrt{nd} + \sqrt{m} + \sqrt{2\gamma n\log n})  \\
& \leq n\|\bar{\BA}\|^2 \left( 1 +  \frac{4C_0\kappa^2 \sigma \sqrt{d} (\sqrt{nd} + \sqrt{m} + \sqrt{2\gamma n\log n})}{\sqrt{n}\|\bar{\BA}\|}  \right).
\end{align*}

As a result, to ensure a benign landscape of~\eqref{def:bm}, it suffices to have
\[
\frac{\lambda_{\max}(\BL)}{\lambda_{d+1}(\BL)} \leq \kappa^2\cdot \frac{\sqrt{n}\|\bar{\BA}\| +C \kappa^4 \sigma \sqrt{d} (\sqrt{nd} + \sqrt{m} + \sqrt{2\gamma n\log n}) }{ \sqrt{n}\|\bar{\BA}\| -  C \kappa^4 \sigma \sqrt{d} (\sqrt{nd} + \sqrt{m} + \sqrt{2\gamma n\log n})} \leq \frac{p+d-2}{2d}
\]
which is equivalent to
\[
\sigma \leq \frac{p + d - 2\kappa^2 d- 2 }{p + d + 2\kappa^2 d - 2} \cdot \frac{\sqrt{n}\|\bar{\BA}\|}{C \kappa^4 \sqrt{d} (\sqrt{nd} + \sqrt{m} + \sqrt{2\gamma n\log n})}.
\]
\end{proof}

%\bibliography{HDKuramoto.bib}
%\bibliographystyle{abbrv}

\end{document}